\documentclass{lmcs}
\usepackage[utf8x]{inputenc}
\usepackage{xcolor}
\usepackage{amssymb}
\usepackage{amsmath}
\usepackage{amsthm}
\usepackage{graphicx}
\usepackage[english]{babel}
\usepackage[numbers,sort&compress]{natbib}   
\usepackage{tikz}
\usepackage{slashed} 
\usepackage{url}
\usepackage{enumerate}
\usepackage[bookmarks=false]{hyperref}
\usepackage{xargs}
\usepackage{MnSymbol}
\usepackage{bbding}
\usepackage{xifthen}
\usepackage{enumitem}
\usepackage{ic}
\usetikzlibrary{arrows}

\newtheorem{example}[thm]{Example}
\newtheorem{definition}[thm]{Definition}
\newtheorem{theorem}[thm]{Theorem}
\newtheorem{lemma}[thm]{Lemma}
\newtheorem{corollary}[thm]{Corollary}

\newtheorem{proposition}[thm]{Proposition}

\begin{document}
\title[Polynomial Path Orders: A Maximal Model]{Polynomial Path Orders: A Maximal Model}

\author[]{Martin Avanzini}
\author[]{Georg Moser}
\address{Institute of Computer Science\\ University of Innsbruck\\ Austria}
\email{\{martin.avanzini,georg.moser\}@uibk.ac.at}
\thanks{The first author partially supported by a grant of the University of Innsbruck.}
\thanks{The second author is partially supported by FWF (Austrian Science Fund) project I-608-N18}

\renewcommand{\labelitemi}{-}

\begin{abstract}
This paper is concerned with the automated complexity analysis of 
\emph{term rewrite systems} (\emph{TRSs} for short)
and the ramification of these in \emph{implicit computational complexity} 
theory (\emph{ICC} for short).
We introduce a novel path order with multiset status, the \emph{polynomial path order} $\gpop$. 
Essentially relying on the principle of 
\emph{predicative recursion} as proposed by Bellantoni and Cook, 
its distinct feature is the tight control of resources on compatible TRSs: 
The (innermost) runtime complexity of compatible TRSs is polynomially bounded. 
We have implemented the technique, as underpinned by our experimental evidence
our approach to the automated runtime complexity analysis is not only feasible, 
but compared to existing methods incredibly fast.

As an application in the context of ICC we provide an order-theoretic 
characterisation of the polytime computable functions. 
To be precise, the polytime computable functions are exactly the functions
computable by an orthogonal constructor TRS compatible with \POPSTAR.
\end{abstract}

\maketitle

\section{Introduction}\label{s:intro}
As a special form of equational logic, \emph{term rewriting} has found many applications in 
automated deduction and verification.
Term rewriting is a conceptually simple but powerful abstract model of
computation that underlies much of declarative programming, and
the automated time complexity analysis of \emph{term rewrite systems} (\emph{TRSs} for short) is of particular interest. 
A natural way to measure the time complexity of a TRS $\RS$ is to measure the
length $\ell$ of derivations
$$
  f(\seq{v}) \rew[\RS] s_1 \to s_2 \cdots \rew[\RS] s_\ell = w
$$
in terms of the sizes of the initial arguments $\seq{v}$.
Maybe surprisingly, this \emph{unitary cost model} is polynomially invariant~\cite{LM09,AM10b}:
the result $w$ of $f(\seq{v})$ can be computed on a conventional model of computation 
in time polynomial in $\ell$. 
\emph{Runtime complexity analysis} is an active research area
in rewriting. 
See~\cite{M09} for a broad overview in this research field.
Since the \emph{feasible} functions are often associated with the
polytime computable functions, estimating polynomial bounds is of particular interest.
Virtually all methods developed in this field go back to termination techniques.
Termination of rewrite systems has been studied extensively, 
and majored to a state where it has become practical
to study the termination of \emph{real world programs} by translations to rewrite systems.
Source languages cover not only functional programs (see for instance~\cite{GRMSST11} that studies \emph{Haskell}),
but also logic  (c.f.~\cite{O01} or~\cite{SSG10} for \emph{Prolog} programs) and 
imperative programs (for \emph{Java${}^\text{\texttrademark}$ bytecode} in~\cite{OBEG10} and recently~\cite{BMOG12}).
This trend is also reflected in the annual termination competition (TERMCOMP)%
\footnote{\url{http://termcomp.uibk.ac.at/}.} that features dedicated categories for all mentioned programming languages.
Verifying that such translations are complexity preserving, 
rewriting can provide a \emph{unified backend} for complexity analysis of programs,
written in different languages and different paradigms.

It is clear that \emph{reduction orders}, for instance \emph{polynomial interpretations}
and \emph{recursive path orders} not only verify termination but also bind the length of reductions. 
For instance, the longest possible rewrite sequence in polynomial terminating TRSs is double-exponentially bounded in the size
of the initial term, cf.~\cite{HL89}. Similar, \emph{multiset path orders} (\emph{MPO} for short) induce primitive recursive 
complexity~\cite{H92}, the induced bound for the \emph{Knuth-Bendix} order is two-recursive~\cite{M06}
and for \emph{lexicographic path orders} it is even \emph{multiply recursive}~\cite{W95}.
In a modern termination prover, these orders play a fundamental role in their combination with transformation 
techniques like \emph{semantic labeling}~\cite{Z95} and the \emph{dependency pair method}~\cite{AG00}.
Based on a careful analysis of the induced derivational complexity~\cite{MS11}, 
Schnabl conjectures
\begin{quote}
  \emph{[t]he derivational complexity of any rewrite system that can be proven 
  terminating using a recent termination prover is bounded
  by a multiply recursive function.}
\end{quote}
With our tool \TCT, the \emph{Tyrolean complexity tool}%
\footnote{$\TCT$ is open source and available from \url{http://cl-informatik.uibk.ac.at/software/tct}.},
we have demonstrated that a termination prover, employing only suitable miniaturised termination techniques, 
can form a powerful complexity analyser. \TCT\ puts special focus on 
proving polynomial bounds on the runtime (respectively derivational) complexity of TRSs. 
However, it is worth emphasising that 
the most powerful techniques for polynomial runtime complexity analysis currently available,
basically employ semantic considerations on the rewrite systems, which are notoriously
inefficient.
We just mention very recently work on a miniaturisation of matrix interpretations due to Middeldorp~et~al.~\cite{MMNWZ11}.
Recent breakthroughs in complexity analysis have also been achieved with 
the development of variations of dependency pairs~\cite{HM08,NEG11,HM11} as well as
modularity results~\cite{HZMK10}.

\subsection{Motivation and Contributions}
To overcome the notorious inefficiency of semantic techniques in
runtime complexity analysis we aim at a syntactic method to analyse
polynomial runtime complexity of rewrite systems. 
A suitable starting point for such an analysis is given by the multiset path order $\MPO$. 
$\MPO$ not only induces primitive recursive bounds on the length of derivations, it
even characterises the primitive recursive functions~\cite{CW97}:
any function computed by an $\MPO$-terminating TRS is primitive recursive, vice versa, 
any primitive recursive function can be stated as an $\MPO$-terminating TRS.

It is well known that the principles of \emph{data tiering} 
introduced by Simmons~\cite{Simmons:1988} and Leivant~\cite{L91}
can be used to characterise small complexity classes like $\FP$ in 
a purely syntactic manner. 
In particular Bellantoni and Cook~\cite{BC92} embodies the principle 
of \emph{predicative recursion}, a form of tiering, on the definition of the primitive recursive 
functions, resulting in a recursion theoretic characterisation of $\FP$.
The here proposed \emph{polynomial path order} (\emph{\POPSTAR} for short), embodies
the principle of predicative recursion onto $\MPO$, with the 
distinctive feature that \POPSTAR\ induces \emph{polynomial} bounds 
on the length of derivations. To motivate this order, 
let us first recapitulate central ideas of~\cite{BC92}.
For each function $f$, the arguments to $f$ are separated into \emph{normal} 
and \emph{safe} ones.
To highlight this separation, we write $f(\svec{x}{y})$ where arguments to the left 
of the semicolon are normal, the remaining ones are safe.
Bellantoni and Cooks define a class $\B$,
consisting of a small set of initial functions 
and that is closed under \emph{safe composition} and \emph{safe recursion on notation} (\emph{safe recursion} for brevity).
The crucial ingredient in $\B$ is that 
a new function $f$ is defined via safe recursion by the equations
\begin{equation}\label{scheme:srn} \tag{\ensuremath{\mathsf{SRN}}}
  \begin{array}{r@{\;}c@{\;}l}
   f(\sn{0,\vec{x}}{\vec{y}}) & = & g(\sn{\vec{x}}{\vec{y}}) \\
   \qquad\quad f(\sn{2z + i,\vec{x}}{\vec{y}}) & = &
   h_i(\sn{z,\vec{x}}{\vec{y},f(\sn{z,\vec{x}}{\vec{y}})})\quad i \in \set{1,2}
  \end{array}
\end{equation}
for functions $g,h_1$ and $h_2$ already defined in $\B$.
Unlike primitive recursive functions, the
stepping functions $h_i$ cannot perform recursion on the \emph{impredicative} value $f(\sn{z,\vec{x}}{\vec{y}})$.
This is a consequence of data tiering. Recursion is performed on normal, 
and recursively computed result are substituted into safe argument position.
To maintain the separation, safe composition restricts the usual composition operator
so that safe arguments are not substituted into normal argument position.
Precisely, for functions $h$, $\vec{r}$ and $\vec{s}$ already defined in $\B$, 
a function $f$ is defined by safe composition using the equation
\begin{equation}\label{scheme:sc} \tag{\ensuremath{\mathsf{SC}}}
  f(\sn{\vec{x}}{\vec{y}}) = h(\sn{\vec{r}(\sn{\vec{x}}{})}{\vec{s}(\sn{\vec{x}}{\vec{y}})}) \tpkt
\end{equation}
Crucially, the safe arguments $\vec{y}$ are absent in normal arguments to $h$.
The main result from~\cite{BC92} states that $\B = \FP$. 

\medskip
Polynomial path orders enforce safe recursion on compatible TRSs. 
In order to employ the separation of normal and safe arguments, 
we fix for each defined symbol
a partitioning of argument positions into \emph{normal} and \emph{safe} positions. 
For constructors we fix that all argument positions are safe. 
Moreover \POPSTARS\ restricts recursion to normal argument.
Dual only safe argument positions allow the substitution of recursive calls.
Via the order constraints we can also guarantee that functions 
are composed in a safe manner.
This syntactic account of predicative recursion delineates a class of rewrite systems: 
a rewrite system $\RS$ is called \emph{predicative recursive} if $\RS$ is
compatible with $\POPSTAR$. 
For motivation consider the TRS $\RSsat$ given in Example~\ref{ex:rssat} the that encodes
the function problem $\FSAT$ associated to the well-known 
satisfiability problem $\SAT$.
Notably $\FSAT$ is complete for the class of \emph{function problems over $\NP$} ($\FNP$ for short),
compare~\cite{Papa}.
\begin{example}
\label{ex:rssat}
The TRS $\RSsat$ is defined as follows. 
A conjunctive normal form is encoded as a list of non-empty clauses, 
clauses being lists of literals, in the obvious way.
Lists are constructed as usual from the constant $\nil$ and the binary constructor ($\cons$).
Literals are encoded as binary strings (build from the $\varepsilon$, $\mZ$ and $\mO$) with 
the most significant bit reserved for its plurality. 
The TRS $\RSsat$ contains a conditional
\begin{align*}
  \mif(\sn{}{\mtrue,t,e}) & \to t & \mif(\sn{}{\mfalse,t,e}) & \to e
\end{align*}
and defines negation
\begin{align*}
  \mneg(\sn{}{\mO(x)}) & \to \mZ(x) & \mneg(\sn{}{\mZ(x)}) & \to \mO(x)
\end{align*}
as well as equality:
\begin{align*}
  \meq(\sn{\mZ(x)}{\mZ(y)}) & \to \meq(\sn{x}{y})
  & \meq(\sn{\mZ(x)}{\mO(y)}) & \to \mfalse & 
  \meq(\sn{\varepsilon}{\varepsilon}) & \to \mtrue\\
  \meq(\sn{\mO(x)}{\mO(y)}) & \to \meq(\sn{x}{y})
  & \meq(\sn{\mO(x)}{\mZ(y)}) & \to \mfalse \tpkt
\end{align*}

A list of literals is \emph{consistent} if an atom does not 
occur positively and negatively.
\begin{align*}
  \verify(\sn{\nil}{}) & \to \mtrue & \verify(\sn{l \cons ls}{}) & \to \mif(\sn{}{\member(\sn{\mneg(\sn{}{l}),ls}{})}, \mfalse, \verify(\sn{ls}{})) \\
  \member(\sn{x,\nil}{}) & \to \mfalse & \member(\sn{x,y \cons ys}{}) & \to \mif(\sn{}{\meq(\sn{x}{y}), \mtrue, \member(\sn{x,ys}{})})
\end{align*}
The computed assignment will be a consistent list of literals.
Note that a satisfying assignment necessarily contains a literal for every clause $c$.
The following rules guess such an assignment and verify whether it is consistent.
\begin{align*}
  \issat(\sn{c}{}) & \to \issat'(\sn{\guess(\sn{c}{})}{}) &
  \issat'(\sn{as}{}) & \to \mif(\sn{}{\verify(\sn{as}{}),as,\unsat}) \\
  \guess(\sn{\nil}{}) & \to \nil & 
  \guess(\sn{c \cons cs}{}) & \to \choice(\sn{c}{}) \cons \guess(\sn{cs}{}) \tpkt
\end{align*}
Here $\choice$ given by the rules 
\begin{align*}
  \choice(\sn{a \cons \nil}{}) & \to a & 
  \choice(\sn{a \cons b \cons bs}{}) & \to a & 
  \choice(\sn{a \cons b \cons bs}{}) & \to \choice(\sn{b \cons bs}{})
\end{align*}
selects nondeterministically an literal from a clause.
This concludes the definition of $\RSsat$. 

It can be verified that $\RSsat$ is compatible with the multiset path order $\gmpo$ 
with underlying precedence $\qp$ satisfying
\begin{alignat*}{3}
\guess & ~~\sp~~ \choice 
& \meq & ~~\sp~~ \mtrue, \mfalse
& \issat & ~~\sp~~ \issat', \guess \\
\member & ~~\sp~~ \mtrue, \mfalse
& \qquad \verify & ~~\sp~~ \mif, \member, \mneg, \mtrue, \mfalse
& \qquad \issat' & ~~\sp~~ \mif, \verify, \unsat \\
{\mneg} & ~~\sp~~ \mZ, \mO
\end{alignat*}
Using the separation of argument positions as indicated in the rules,
where in the spirit of $\B$ constructors admit only safe arguments, 
we can even prove compatibility with $\gpop$ based on the same precedence, i.e., 
$\RSsat$ is predicative recursive.
\end{example}
Note that $\RSsat$ does not rigidly follow 
safe recursion~\eqref{scheme:srn} and safe composition \eqref{scheme:sc}.
Notably values are formed from an arbitrary algebra and are not restricted to words.
Also $\gpop$ allows in principle arbitrary deep right-hand sides.
Still the main principle, namely prohibition of recursion on impredicative values, 
remains reflected.
In total, we establish following results.
\begin{description}[leftmargin=0.3cm]
\item[{Automated Runtime Complexity Analysis of TRSs}]
  We establish that for predicative recursive TRSs $\RS$, 
  the\emph{ (innermost) runtime complexity function} is polynomially bounded.
  To the best of our knowledge, the polynomial path order is the first purely syntactic 
  approach that establishes feasible bound on the runtime complexity of TRSs.
  We have implemented the here proposed techniques in \TCT.\@
  The experimental evidence obtained indicates the viability of the method.
  
  For the predicative recursive TRS $\RSsat$ from Example~\ref{ex:rssat}
  this result implies that the number of rewrite steps starting from $\issat(\sn{c}{})$ is polynomially bounded in
  the size of the CNF $c$. This can even be automatically verified\footnote{
    To our best knowledge $\TCT$ is currently the only complexity
    tool that can provide a complexity certificate for the TRS $\RSsat$,
    compare~\url{http://termcomp.uibk.ac.at}.}.
  Due to the polynomial invariance theorem~\cite{AM10b} we can thus 
  that $\FSAT$ belongs to $\FNP$.
  
\item[{Resource free characterisation of $\FP$}]
  The class of predicative recursive rewrite systems
  entail new \emph{order-theoretic} characterisation of $\FP$, the \emph{polytime computable functions}.
  This bridges the gap to \emph{implicit computational complexity} (\emph{ICC} for short) theory. 

  \POPSTAR\ is \emph{sound} for $\FP$, i.e., (confluent and) predicative recursive TRSs
  compute only polytime computable functions.
  Moreover we can also prove that predicative recursive TRSs are \emph{complete} for $\FP$, 
  in the sense that every polytime computable function $f$
  is defined by a (orthogonal and) predicative recursive TRS $\RS_f$.

\item[{Parameter Substitution}]
  We extend upon \POPSTAR\ by proposing 
  a generalisation \POPSTARP, admitting the same properties as outlined above, 
  but that allows to handle more general recursion schemes that make
  use of parameter substitution.
  As a corollary to this and the fact that the runtime complexity
  of a TRS forms an invariant cost model we conclude a non-trivial
  closure property of Bellantoni and Cooks definition of the feasible functions.
\end{description}

The present article collects our ongoing work on polynomial path orders. 
The order $\POPSTAR$ has been introduced first in~\cite{AM08}, extended to
quasi-precedences in~\cite{AMS08} and the extension $\POPSTARP$
appeared first in the Workshop on Termination of 2009~\cite{AM09b}.
Apart from the usual corrections of technicalities, we make here the following new contributions:
\begin{itemize}
\item
  The presented definition of $\POPSTAR$ is more liberal and captures predicative recursion more precisely, 
  compare~\cite[Definition~4]{AM08} and Definition~\ref{d:gpop} from Section~\ref{s:popstar}.
\item
  To show that $\POPSTAR$ is sound for $\FP$, we relied in~\cite{AM08}
  on a certain typing of constructors that guaranteed that sizes of values are polynomial in 
  their depth. In particular, the typing prohibited tree structures a~priori. 
  Our new soundness result (c.f Theorem~\ref{t:icc:soundness} from Section~\ref{s:icc}) is more general and permits arbitrary values.

\item The propositional encoding used in our automation of polynomial path orders (c.f. Section~\ref{s:impl})
  has been considerably overhauled.
\end{itemize}

\subsection{Related Work}
There are several accounts of predicative analysis of recursion in the (ICC) literature. 
We mention only those related works which are directly comparable to our work. See~\cite{BMR09} for an 
overview on ICC.\@
The mental predecessor of $\POPSTAR$ is the \emph{path order for $\FP$} 
as put forward in~\cite{AM05}. 
Our main motivation lies in the observation that this order is directly only applicable to a handful of simple TRSs.
This order only gains power if addition transformations are performed. 
But unfortunately powerful transformations are difficult to find automatically.

Notable the clearest connection of our work is to Marion's \emph{light multiset path order} (\emph{LMPO} for short)~\cite{M03}. 
This path order forms a strict extension of the here 
proposed order $\POPSTAR$. Similar to $\POPSTAR$ it characterises $\FP$. 
As exemplified below however, 
compatible TRSs do not admit polynomially bounded runtime complexity in general.
This renders $\LMPO$ non-usable in our complexity analyser $\TCT$.
The definition of $\POPSTAR$ has been calibrated with some effort to prevent such behaviour.
\begin{example}
The TRS $\RSbin$ is given by the following rules:
\label{ex:RS2}
\begin{align*}
    \bin(\sn{x,\Null}{}) & \to \ms(\Null) & 
    \bin(\sn{\Null,\ms(y)}{}) & \to \Null &
    \bin(\sn{\ms(x),\ms(y)}{}) & \to \mP(\sn{}{\bin(\sn{x,\ms(y)}{}),\bin(\sn{x,y}{})})
  \end{align*}
For a precedence $\qp$ that fulfils $\bin \sp \ms$ and $\bin \sp \mP$ we obtain 
that $\RSbin$ is compatible with $\LMPO$. 
However it is straightforward to verify that the family of terms
$\bin(\ms^n(\Null),\ms^m(\Null))$ admits (innermost) derivations whose length grows exponentially in $n$.
Still the underlying function can be proven polynomial, essentially relying on memoisation techniques, c.f.~\cite{M03}.
\end{example}

The result of our main theorem can also be obtained
if one considers polynomial interpretations, where the interpretations of
constructor symbols is restricted. Such restricted polynomial interpretations
are called \emph{additive} in~\cite{BCMT01}. Note that additive polynomial interpretations
also characterise the functions computable in polytime, cf.~\cite{BCMT01}.
Although incomparable to our technique, 
unarguably such semantic techniques admit a better intensionality, 
but are difficult to implement efficiently in an automated setting. 
In our complexity tool \TCT, we see \POPSTAR\ as a fruitful and fast extension that handles
systems in a fraction of a second.

We also want to mention recent approaches for the automated analysis of 
resource usage in programs.
Notably, Hoffmann~et~al.~\cite{HAH11} provide an automatic multivariate amortised 
cost analysis exploiting typing, which extends earlier results on amortised 
cost analysis.
To indicate the applicability of our method we have employed a straightforward (and complexity preserving) 
transformation of the RAML programs considered
in~\cite{HAH11} into TRSs. Equipped with \POPSTAR\ our complexity
analyser \TCT\ can handle all examples from~\cite{HAH11}.
Finally Albert et al.~\cite{AAGGPRRZ:2009} present an automated complexity tool
for Java${}^\text{\texttrademark}$ Bytecode programs and
Gulwani~et~al.~\cite{GMC09} as well as Zuleger~et~al.~\cite{ZulegerGSV11} 
provide an automated complexity tool for C programs.

\subsection{Outline}
The remainder of this paper is organised as follows.
In the next section we recall basic notions and starting points of
this paper.
In Section~\ref{s:popstar} we introduce polynomial path orders along with our main result.
In the subsequent Sections~\ref{s:pint}--\ref{s:embed} we show that the (innermost) 
runtime-complexity of predicative recursive TRSs is polynomially bounded:
in Section~\ref{s:pint} we set the stage by introducing a notion of \emph{predicative interpretation};
in Section~\ref{s:pop} we present an extended version of the aforementioned path order on sequences~\cite{AM05}, 
and we show that our extension is still sound (c.f. Corollary~\ref{c:pop});
section~\ref{s:embed} finally shows that predicative interpretations embed derivations into the order on sequences, 
establishing our central argument.

In Section~\ref{s:icc} we then present our ramification of polynomial path orders in ICC.\@
Parameter substitution is incorporated in Section~\ref{s:popstarps}.
Our implementation is detailed in Section~\ref{s:impl} and ample experimental evidence 
is provided in Section~\ref{s:exps}. 
Finally, we conclude and present future work in Section~\ref{s:conclusion}.


\section{Preliminaries}\label{s:basics}
We denote by $\N$ the set of natural numbers $\{0,1,2,\dots\}$.
Let $R$ be a binary relation.
The transitive closure of $R$ is denoted by $R^+$ and
its transitive and reflexive closure by $R^{\ast}$. 
For $n\in \N$ we denote by $R^n$ the \emph{$n$-fold composition of $R$}.
The binary relation $R$ is \emph{well-founded} if 
there exists no infinite chain $a_0, a_1, \dots$ with $a_i \mathrel{R} a_{i+1}$
for all $i \in \N$. Moreover, we say that $R$ is well-founded on a set $A$ if 
there exists no such infinite chain with $a_0 \in A$.
The relation $R$ is \emph{finitely branching} if for all elements $a$, the set $\{b \mid a \mathrel{R} b\}$ is finite.

A \emph{proper order} is an irreflexive and transitive binary relation.
A \emph{preorder} is a reflexive and transitive binary relation. 
An \emph{equivalence relation} is reflexive, symmetric and transitive.
For a preorder $\succcurlyeq$, we denote the induced equivalence relation by $\eqi$ and induced proper order  
by $\succ$. 

A multiset is a collection in which elements are allowed to
occur more than once. We denote by $\msetover(A)$ the set of multisets over $A$
and write $\mset{a_1,\dots,a_n}$ to denote multisets with elements $a_1, \dots,a_n$.
We use $m_1 \uplus m_2$ for the summation and $m_1 \backslash m_2$ for difference on multisets $m_1$ and $m_2$.
The \emph{multiset extension} $\mextension{R}$ \emph{of a relation $R$ on $A$} is the 
relation on $\msetover(A)$ such that $M_1 \mextension{R} M_2$ if there exists 
multisets $X,Y \in \msetover(A)$ satisfying 
\begin{enumerate}    
    \item $M_2 = (M_1 \backslash X) \uplus Y$, 
    \item $\varnothing \not= X \subseteq M_1$ and 
    \item for all $y \in Y$ there exists an element $x \in X$ such that $x \mathrel{R} y$.
\end{enumerate}
In order to cleanly extend this definition to preorders and equivalences, we follow~\cite{Ferreira95}.
Let $\eqi$ denote an equivalence relation over the set $A$ and let ${\succcurlyeq} = {\succ} \cup {\eqi}$ be a binary
relation over $A$ so that $\succ$ and $\eqi$ are \emph{compatible} in the following sense: 
${{\eqi} \cdot {\succ} \cdot {\eqi}} \subseteq {{\succ}}$. 
Let $\eclass{a}$ denotes the \emph{equivalence class of $a \in A$} with respect to $\eqi$.
By the compatibility requirement, 
the extension $\sqsupset$ of $\succ$ to equivalence classes 
such that $\eclass{a}\sqsupset\eclass{b}$ if and only if $a \succ b$, 
is well defined.
We define the \emph{strict multiset extension $\mextension{\succ}$ of $\succcurlyeq$} as 
${M_1} \mextension{\succ} {M_2}$ if and only if ${\eclass{M_1}} \mextension{\sqsupset} {\eclass{M_2}}$.
Further, the \emph{weak multiset extension $\mextension{\succcurlyeq}$ of $\succcurlyeq$} is given 
by ${M_1} \mextension{\succcurlyeq} {M_2}$ if and only if ${\eclass{M_1}} \mextension{\sqsupset} {\eclass{M_2}}$
or ${\eclass{M_1}} = {\eclass{M_2}}$ holds. 
Note that if ${\succcurlyeq}$ is a preorder (on $A$) then
$\mextension{\succ}$ is a proper order
and $\mextension{\succcurlyeq}$ a preorder on $\msetover{A}$, cf.~\cite{Ferreira95}. 
Also $\mextension{\succ}$ is well-founded if $\succ$ is well-founded. 

\subsection{Complexity Theory}
We assume modest familiarity in \emph{complexity theory}, notations 
are taken from~\cite{Papa}. 
The \emph{functional problem} $F_R$ associated with an binary relation $R$
is: given $x$ find some $y$ such that $(x,y) \in R$ holds if $y$ exists, 
otherwise return $\m{no}$. 
A binary relation $R$ on words is called \emph{polynomial balanced}
if for all $(x,y) \in R$, the size of $y$ is polynomially bounded in $x$.
The relation $R$ is \emph{polytime decidable} if $(x,y) \in R$ 
is decided by a deterministic Turing machine (TM for short) $M$ operating in polynomial time.
The class $\NP$ is the class of languages $L$ admitting polynomially balanced, polytime decidable 
relations $R_L$~\cite[Chapter~9]{Papa}: $L =  \{x \mid (x,y) \in R_L \text{ for some $y$}\}$.
The class $\FNP$ is the class of \emph{function problems associated with $\NP$} in the above way. 
The class of \emph{polynomial time computable functions} $\FP$ (\emph{polytime computable} for short) 
is the subclass resulting if we only consider function problems in $\FNP$ that 
can be solved in polynomial time~\cite[Chapter 10]{Papa}.

We say that a function problem $F$ reduces to a function problem $G$ if there 
exist functions $s$ and $r$, both computable in logarithmic space, such that 
for all $x,y$ with $F$ computing $y$ on input $x$, 
$G$ computes on input $s(x)$ the output $z$ with $r(z) = y$.
Note that both $\FNP$ and $\FP$ are closed under reductions.
We also note that nondeterministic Turing machines running in polynomial time compute function 
problems from $\FNP$.
\begin{proposition}\label{p:fnp}
  Let $N$ be a nondeterministic Turing machine that computes the function problem $F$ in polynomial time.
  Then $F \in \FNP$.
\end{proposition}
\begin{proof}
  Define the following relation $R$: $(x,y) \in R$ if and only if $y$ is the encoding of an accepting 
  computation of $N$ on input $x$.
  Since $N$ operates in polynomial time, $R$ is polynomially balanced, as 
  it can be checked in linear time that $y$ encodes an accepting run of $N$ on input $x$, 
  $R$ is polytime decidable.
  Hence the functional problem $F_R$ that computes an accepting runs $y$ of $N$ on input $x$ is in $\FNP$.
  Finally notice that $F$ reduces to $F_R$.
  To see this, employ following reduction:
  the function $r$ is simply the identity function; the logspace computable function $s$ extracts the result of $N$ on input $x$ 
  from the accepting run $y$ computed by $F_R$ on input $x$.
  We conclude the lemma since $\FNP$ is closed under reductions.
\end{proof}

\subsection{Term Rewriting}
We assume at least nodding acquaintance with the basics of term rewriting~\cite{BN98}.
We fix the bare essential of
notions and notation that are used in the paper.

Throughout the paper, we fix a countably infinite set of \emph{variables} $\VS$
and a finite \emph{signature} $\FS$ of \emph{function symbols}.
The signature $\FS$ is partitioned into \emph{defined symbols} $\DS$ and 
\emph{constructors} $\CS$.
The set of \emph{values}, \emph{basic terms} and
\emph{terms} is defined according to the grammar
\begin{alignat*}{4}
  & \text{\textsf{(Values)}} & & \quad & \Val & \ni v &~&\defsym~x \mid c(\seq{v}) \\
  & \text{\textsf{(Basic Terms)}} & & & \BASICS & \ni s & &  \defsym~ x \mid f(\seq{v}) \\
  & \text{\textsf{(Terms)}} & & & \TERMS & \ni t & & \defsym~ x \mid c(\seq{t}) \mid f(\seq{t})
\end{alignat*}
where $x \in \VS$, $c \in \CS$, and $f \in \DS$.

The \emph{arity} of a function symbol $f \in \FS$ is denoted by $\ar(f)$.
We write $s \superterm t$ if $t$ is a \emph{subterm} of $t$, the 
strict part of $\superterm$ is denoted by $\supertermstrict$.
The \emph{size} of a term $t$ is denoted by $\size{t}$ and refers
to the number of variables and function symbols contained in $t$.
We denote by $\depth(t)$ the \emph{depth} of $t$ which is
defined as $\depth(t) = 1$ if $t \in \VS$ and $\depth(f(\seq{t})) = 1 + \max\{\depth(t_i) \mid i = 1,\dots n\}$.
Here we employ the convention that the maximum of an empty set is equal to $0$.

Let $\qp$ be a preorder on the signature $\FS$, called \emph{quasi-precedence}
or simply \emph{precedence}. We always write $\sp$ for the 
induced proper order and 
$\ep$ for the induced equivalence on $\FS$.
We lift the equivalence ${\ep}$ to terms 
modulo argument permutation:
$s \eqi t$ if either $s = t$ or
$s = f(\seq{s})$ and $t = g(\seq{t})$ where $f \ep g$
and for some permutation $\pi$,
$s_i \eqi t_{\pi(i)}$ for all $i \in \{1,\dots,n\}$.
Further we write $s \esuperterm t$ if $t$ is a subterm
of $s$ modulo $\eqi$, i.e.,  $s~{\superterm} \cdot {\eqi}~{t}$.
We denote by $\sigbelow{f}{\FS} \defsym \{g \mid f \sp g\}$ the set of function symbols 
below $f$ in the precedence $\qp$.
This notion is extended to sets $F \subseteq \FS$ by 
$\sigbelow{F}{\FS} \defsym \bigcup_{f \in F} \sigbelow{f}{\FS}$.
The \emph{rank} of a function symbol is inductively defined by
$\rk(f) = \max\{1 + \rk(g) \mid f \sp g\}$.

A \emph{rewrite rule} is a pair $(l,r)$ of terms, in notation $l \to r$, 
such that $l$ is not a variable and all variables in $r$ occur also in $l$.
Here $l$ is called the \emph{left-hand}, and $r$ the \emph{right-hand side} of $l \to r$.
A \emph{term rewrite system} (\emph{TRS} for short) $\RS$ over
$\TERMS$ is a set of \emph{rewrite rules}.
In the following, $\RS$ always denotes a TRS.\@
If not mentioned otherwise, we assume $\RS$ is \emph{finite}.
A relation on $\TERMS$ is a \emph{rewrite relation} if it is
closed under contexts and closed under substitutions. 
The smallest rewrite relation that contains $\RS$ is denoted by
$\rew$. 

A term $s \in \TERMS$ is called a \emph{normal form} if there is no
$t \in \TERMS$ such that $s \rew t$. 
With $\NF(\RS)$ we denote the set of all normal forms of a TRS $\RS$.
Whenever $t$ is a normal form of $\RS$ we write $s \rsn t$ for $s \rss t$.
The \emph{innermost rewrite relation}, denoted as $\irew$, is the restriction 
of $\rew$ where arguments are normal forms.
The TRS $\RS$ is \emph{terminating} if no infinite rewrite sequence exists.
The TRS $\RS$ has \emph{unique normal forms} if for all 
$s, t_1, t_2 \in \TERMS$ with $s \rsn t_1$ and 
$s \rsn t_2$ we have $t_1 = t_2$.
The TRS $\RS$ is called \emph{confluent} if for all $s, t_1, t_2 \in \TERMS$
with $s \rss t_1$ and $s \rss t_2$ there exists a term $u$ such that
$t_1 \rss u$ and $t_2 \rss u$. An \emph{orthogonal} TRS is 
a left-linear and non-overlapping TRS~\cite{BN98}. Note that
every orthogonal TRS is confluent. 
The TRS $\RS$ is a \emph{constructor} TRS if all left-hand sides are basic terms.
A defined function symbol is \emph{completely defined} (with respect to $\RS$) 
if it does not occur in any term in normal form, i.e., 
functions are reducible on all terms. 
The TRS $\RS$ is \emph{completely defined} if each defined symbol 
is completely defined.

\subsection{Rewriting as Computational Model}
We fix \emph{call-by-value} semantics and only consider 
\emph{constructor} TRSs $\RS$. Input and output are taken from the set of 
values $\Val$, and defined symbols $f \in \DS$ denote computed functions. 
More precise, a (finite) \emph{computation} of $f\in \DS$ on input $\seq{v} \in \Val$
is given by \emph{innermost} reductions
$$
  f(\seq{v}) = t_0 \irew t_1 \irew \cdots \irew t_\ell = w \tpkt
$$
If the above computation ends in a value, i.e., $w \in \Val$, 
we also say that $f$ \emph{computes} on input $\seq{v}$ in $\ell$ steps the value $w$.
To also account for nondeterministic computation, 
we capture semantics of $\RS$ by assigning to each $n$-ary defined symbol 
$f \in \DS$ an $n+1$-ary relation $\sem{f}$ that maps 
input arguments $\seq{v}$ computed values $w$.
A \emph{finite} set $\NA$ of \emph{non-accepting patterns} is used to distinguish meaningful outputs
$w$ from outputs that should not be considered part of the computation.
A value $w$ \emph{is accepting} with respect to $\NA$ 
if there exists no $p \in \NA$ and substitution $\sigma$ such that $p\sigma = w$ holds.
A typical example of a meaningful value that should not be accepted is the constant $\unsat$ 
appearing in the TRS $\RSsat$ from Example~\ref{ex:rssat}.
Below functional problem are extended to $n+1$-ary relations in the obvious way.
\begin{definition}
\label{d:computation}
Let $\NA$ be a set of non-accepting patterns.
For each $n$-ary symbol, $f \in \DS$ the TRS $\RS$
\emph{the relation $\sem{f} \subseteq {\Val^{n+1}}$} defined by $f$
is given by
\begin{equation*}
{(\seq{v},w) \in \sem{f}} \quad\defiff\quad {\m{f}(\seq{v}) \irsn[\RS] w} \text{ and $w$ is accepting}\tpkt
\end{equation*}
We say that $\RS$ \emph{computes} the functional problems associated with $\sem{f}$.
\end{definition}

Note that if $\RS$ is confluent, then $\sem{f}$ is in fact a (partial) function. 
Following~\cite{HM08,AM10b} we adopt an \emph{unitary cost model} for rewriting, 
where each reduction step accounts for one time unit. 
Reductions are of course measured in the size of the input. 
\begin{definition}
The\emph{ (innermost) runtime complexity function} $\ofdom{\rc[\RS]}{\N \to \N}$
relates sizes of basic terms $f(\seq{v}) \in \BASICS$ to the maximal 
length of computation. Formally
$$
 \rc[\RS](n) \defsym 
 \max\{\ell \mid \exists s \in \BASICS, \size{s} \leqslant n \text{ and } f(\seq{v})  = t_0 \irew t_1 \irew \dots \irew t_\ell\} \tpkt
$$
\end{definition}
The restriction $s \in \BASICS$ accounts for the fact that computations start only from basic terms. 
We sometimes use $\dheight(s) \defsym \max\{\ell \mid \exists t.~s \irsl{\ell} t\}$
to refer to the \emph{derivation height} of a single term $s$.
Note that the runtime complexity function is well-defined if $\RS$ is \emph{terminating}, 
i.e., $\irew$ is well-founded.
If $\rc[\RS]$ is asymptotically bounded from above by a linear, quadratic,\dots,  polynomial function, 
we simply say that the runtime of $\RS$ is linear, quadratic,\dots, or respectively polynomial.
By folklore it is known that rewriting can be implemented with only polynomial overhead
if terms grow only polynomial during reductions. 

In~\cite{AM10b} we have shown that the unitary cost model is reasonable 
for full rewriting (the deterministic case was proven independently in~\cite{LM09,AM10} using essentially the same approach).
It is not difficult to see that the central Lemma~\cite[Lemma 5.9]{AM10b} that estimates the implementation 
cost of a single rewrite step can be specialised to innermost rewriting. 
We obtain following proposition by specialising~\cite[Theorem~6.2]{AM10b} to innermost rewriting.
\begin{proposition}\label{p:invariance:1}
  Let $\RS$ be a TRS whose is at least linear. 
  There exists a polynomial $p_\RS$ such that for any $f(\seq{v}) \in \BASICS$ of 
  size up to $n$, 
  \begin{enumerate}
  \item any normal form of $f(\seq{v})$ can be computed on a Turing machine in nondeterministic time 
    $p_\RS(\rc[\RS](n))$; and
  \item some normal form of $f(\seq{v})$ is computable on a Turing machine in deterministic time
    $p_\RS(\rc[\RS](n))$.
  \end{enumerate}
\end{proposition}\label{p:invariance}
\noindent Hence there are no surprises here. 
By Proposition~\ref{p:fnp} and Proposition~\ref{p:invariance:1} we obtain:
\begin{proposition}
Let $\RS$ be a rewrite system with polynomial runtime.
Then the functional problems associated with $\sem{f}$ defined by $\RS$ are contained in $\FNP$. 
If $\RS$ is confluent, i.e.\ deterministic, then $\sem{f}$ is a (partial) function contained in $\FP$.
\end{proposition}

\medskip

Our choice of adopting call-by-value semantics is rested 
in the observation that the unitary cost model of unrestricted rewriting 
often overestimates the runtime complexity of computed functions. 
This has to do with the unnecessary duplication of redexes.
\begin{example}\label{ex:dup}
  Consider the constructor TRS $\RSdup$ given by the following rules.
  \begin{alignat*}{6}
    \rlbl{1}: &~& \m{btree}(\sn{\mZ}{}) & \to \m{leaf}  \qquad &
    \rlbl{2}: &~& \m{btree}(\sn{\ms(n)}{}) & \to \m{dup}(\sn{}{\m{btree}(n)})  \qquad &
    \rlbl{3}: &~& \m{dup}(\sn{}{t}) & \to \m{node}(t,t) \tpkt
  \end{alignat*}
  Then for $n \in \N$, $\m{btree}(\ms^n(\mZ))$ computes a binary tree of height $n$
  in a linear number of steps.
  On the other hand, $\RSdup$ gives also rise to a non-innermost reduction
  $$
  \m{btree}(\sn{\ms^{n}(\mZ)}{}) 
  \rew \m{dup}(\m{btree}(\sn{\ms^{n-1}(\mZ)}{}))
  \rew \m{node}(\m{btree}(\sn{\ms^{n-1}(\mZ)}{}), \m{btree}(\sn{\ms^{n-1}(\mZ)}{}))
  \rew \dots
  $$
  obtained by preferring $\m{dup}$ over $\m{btree}$.
  The length of the derivation is however exponential in $n$.
\end{example}
\noindent 
By Proposition~\ref{p:invariance} we obtain $\sem{\m{btree}} \in \FP$.
As indicated later, our analysis can automatically classify the function $\sem{\m{btree}}$ as feasible.


\section{The Polynomial Path Order}\label{s:popstar}
We arrive at the formal definition of \emph{polynomial path order} (\emph{\POPSTAR} for short).
Variants of the here presented definition have been presented in earlier conference publications, 
see~\cite{AM08,AMS08,AM09,AM09b}. 

The order \POPSTAR\ essentially embodies the \emph{predicative analysis} of recursion
set forth by Bellantoni and Cook~\cite{BC92}. 
In \POPSTAR, the separation of argument positions is taken into account in the notion of \emph{safe mapping}.
\begin{definition}\label{d:safemapping}
  A \emph{safe mapping} $\safe$ is a function 
  $\ofdom{\safe}{\FS \to 2^\N}$ that associates 
  with every $n$-ary function symbol $f$ the set of \emph{safe argument positions} 
  $\{i_1, \dots , i_m\} \subseteq \{1,\dots,n\}$.
  Argument positions included in $\safe(f)$ are called \emph{safe},
  those not included are 
  called \emph{normal} and collected in $\normal(f)$.
  For $n$-ary constructors $c$ 
  we require that all argument positions are safe, 
  i.e., $\safe(c) = \{1,\dots,n\}$.
\end{definition}
We refine term equivalence so that the safe mapping is taken into account. 
\begin{definition}\label{d:eqis}
  Let ${\qp}$ denote a precedence and $\safe$ a safe mapping.
  We define \emph{safe equivalence} $\eqis$ for terms $s,t \in TERMS$
  inductively as follows:
  $s \eqis t$ if either $s = t$ or
  $s = f(\seq{s})$, $t = g(\seq{t})$, $f \ep g$
  and there exists a permutation $\pi$ such that for all $i \in \{1,\dots,n\}$, 
  $s_i \eqis t_{\pi(i)}$ and $i \in \safe(f)$ if and only if $\pi(i) \in \safe(g)$.
\end{definition}
To avoid notational overhead, we suppose that for each $k+l$ ary function symbol 
$f$, the first $k$ argument positions are normal, and the remaining 
argument positions are safe, i.e., $\safe(f) = \{k+1,\dots,k+l\}$.
This allows use to write $f(\pseq{s})$
where the separation of safe from normal arguments is directly indicated in terms. 

Let $\qp$ denote a quasi-precedence. We require that the precedence adheres the partitioning of 
$\FS$ into defined symbols and constructors in the following sense.
Then in particular $\eqis$ preserves values, i.e., 
if $s \in \Val$ and $s \eqis t$ then also $t \in \Val$. 
\begin{definition}
A precedence $\qp$ is \emph{admissible} (for \POPSTAR) if $f \ep g$ implies 
that either both $f$ and $g$ are defined symbols, or both are constructors.
\end{definition}
\noindent The following definition introduces an auxiliary order $\gsq$, 
the full order $\gpop$ is then presented in Definition~\ref{d:gpop}.
\begin{definition}\label{d:gsq}
  Let ${\qp}$ denote a precedence and $\safe$ a safe mapping. 
  Consider terms $s, t \in \TERMS$ such that $s = f(\pseq[k][l]{s})$.
  Then $s \gsq t$ if one of the following alternatives holds:
  \begin{enumerate}
  \item\label{d:gsq:st} $s_i \geqsq t$ for some $i \in \{1,\dots,k+l\}$ and, 
    if $f \in \DS$ then $i$ is a normal argument position ($i \in \{1,\dots,k\}$);
  \item\label{d:gsq:ia} $f \in \DS$, $t = g(\pseq[m][n]{t})$ where $f \sp g$ 
    and $s \gsq t_i$ for all $i = 1,\dots,m+n$.
  \end{enumerate}
  Here we set ${\geqsq} \defsym {\gsq} \cup {\eqis}$.
\end{definition}
Consider a function $f$ defined by safe composition from $r$ and $s$, compare scheme~\eqref{scheme:sc}.
The purpose of this auxiliary order is to embody safe composition in the full order $\gpop$.
Note that the auxiliary order can orient $f(\sn{\vec{x}}{\vec{y}}) \gsq r(\sn{\vec{x}}{})$ for defined symbol $f$
with $f \sp r$. 
On the other hand, $f(\sn{\vec{x}}{\vec{y}})$ and safe arguments $y_i$ are incomparable,
and consequently the orientation of $f(\sn{\vec{x}}{\vec{y}})$ and $s(\sn{\vec{x}}{\vec{y}})$ fails.

\begin{definition}\label{d:gpop}
  Let ${\qp}$ denote a precedence and $\safe$ a safe mapping. 
  Consider terms $s, t \in \TERMS$ such that $s = f(\pseq[k][l]{s})$.
  Then $s \gpop t$ if one of the following alternatives holds:
  \begin{enumerate}
  \item\label{d:gpop:st} $s_i \geqpop t$ for some $i \in \{1,\dots,k+l\}$, or
  \item\label{d:gpop:ia} $f \in \DS$, $t = g(\pseq[m][n]{t})$ where $f \sp g$ 
    and the following conditions hold:
    \begin{itemize}
    \item $s \gsq t_j$ for all normal argument positions $j = 1,\dots,m$;
    \item $s \gpop t_j$ for all safe argument positions $j = m+1,\dots,m+n$;
    \item $t_j \not\in \TA(\sigbelow{\Fun(s)}{\FS},\VS)$ for at most one safe argument position $j \in \{m+1,\dots,m+n\}$;
    \end{itemize}
  \item\label{d:gpop:ep} $f \in \DS$, $t = g(\pseq[m][n]{t})$ where $f \ep g$
    and the following conditions hold:
    \begin{itemize}
    \item $\mset{\seq[k]{s}} \gpopmul \mset{\seq[m]{t}}$;
    \item $\mset{\seq[k+l][k+1]{s}} \geqpopmul \mset{\seq[m+n][m+1]{t}}$.
    \end{itemize}
  \end{enumerate}
  Here ${\geqpop} \defsym {\gpop \cup \eqis}$.
\end{definition}

\noindent We say a constructor TRS $\RS$ is \emph{predicative recursive} if
$\RS$ is compatible with an instance $\gpop$ with underlying admissible precedence.

We use the notation $\caseref{\gpop}{i}$ and respectively $\caseref{\gsq}{i}$ 
to refer to the \nth{$i$} case in Definition~\ref{d:gsq} respectively Definition~\ref{d:gpop}.
We emphasise that $\gpop$ is \emph{blind} on constructors, in particular $\gpop$ collapses 
to the subterm relation (modulo equivalence) on values. 
\begin{lemma}\label{l:gpop:val}
  Suppose the precedence underlying $\gpop$ is admissible.
  If $s \gpop t$ and $s \in \Val$ then $s \esupertermstrict t$, 
  in particular $t \in \Val$.
\end{lemma}
The case \cpop{ia} accounts for definitions by safe composition \eqref{scheme:sc}.
The final restriction put onto \cpop{ia} is used to prevent multiple recursive calls
as indicated in Example~\ref{ex:RS2}.
We remark that restrictions put onto \cpop{ia} are weaker compared to the corresponding clause given 
in~\cite[Definition 4]{AM08}%
\footnote{The early definition from~\cite[Definition~4]{AM08}, 
used the full order $\gpop$ only on one argument of the right-hand side (the one that possibly holds the recursive call), 
the remaining arguments were all oriented with the auxiliary order $\gsq$. 
To retain completeness, in~\cite{AM08} we allowed also the 
admittedly ad hoc use of a subterm comparison on safe arguments.}.
The case \cpop{ep} restricts the corresponding case in \MPO\ 
by taking the separation of normal and safe argument positions into account. 
Note that here normal arguments need to decrease. This reflects that as in \eqref{scheme:srn} 
recursion is performed on normal argument positions.
We arrive at the central theorem of this paper.
\begin{theorem}\label{t:popstar}
  Let $\RS$ be predicative recursive TRS.\@
  Then the innermost derivation height of any basic term 
  $f(\svec{u}{v})$ is bounded by a polynomial in the 
  maximal depth of normal arguments $\vec{u}$.
  The polynomial depends only on $\RS$ and the signature $\FS$.
  In particular, the runtime complexity of $\RS$ is polynomial.
\end{theorem}

The proof of Theorem~\ref{t:popstar} is rather involved, and outlined at the end of this 
section. The formal proof is then presented in the subsequent Sections~\ref{s:pint}--\ref{s:embed}.
We clarify first Definition~\ref{d:gpop} on several examples. 
\begin{example}\label{ex:mult1}
  Consider the TRS $\RSmul$ expressing multiplication in Peano arithmetic.
  \begin{alignat*}{4}
    \rlbl{1}: &~& +(\sn{\mZ}{y}) & \to y &
    \rlbl{3}: &~& \times(\sn{\mZ,y}{}) & \to \mZ \\
    \rlbl{2}: &~& +(\sn{\ms(\sn{}{x})}{y}) & \to \ms(\sn{}{+(\sn{x}{y})}) \qquad &
    \rlbl{4}: &~& \times(\sn{\ms(\sn{}{x}),y}{}) & \to +(\sn{y}{\times(\sn{x,y}{})})
  \end{alignat*}
  The TRS $\RSmul$ is predicative recursive, using the precedence
  ${\times}~\sp~{+}~\sp~{\ms}$ and the safe mapping as indicated in the rules:
  The rules $\rlbl{1}$ and $\rlbl{3}$ are oriented by \cpop{st}.
  The rule $\rlbl{3}$ is oriented by $\cpop{ia}$ using ${+}~\sp~{\ms}$
  and $+(\sn{\ms(\sn{}{x})}{y}) \cpop{ep} +(\sn{x}{y})$. 
  Note that the latter enforces that the first argument to $+$ is normal.
  Similar, the final rule $\rlbl{4}$ is oriented by \cpop{ia}, employing
  ${\times}\sp{+}$ together with $\times(\sn{\ms(\sn{}{x}),y}{}) \csq{st} y$ 
  and $\times(\sn{\ms(\sn{}{x}),y}{}) \cpop{ep} \times(\sn{x,y}{})$.
  Note that the latter two inequalities require that 
  the both argument positions of $\times$ are normal, i.e.,\ are 
  used for recursion.
\end{example}  
\begin{example}
  Now consider the extension of $\RSmul$ from Example~\ref{ex:mult1} by the two rules
  \begin{alignat*}{4}
    \rlbl{5}: &~& \m{exp}(\mZ,y) & \to \ms(\sn{}{\mZ})  \qquad\qquad &
    \rlbl{6}: &~& \m{exp}(\ms(\sn{}{x}),y) & \to \times(\sn{y,\m{exp}(x,y)}{})
  \end{alignat*}
  that express exponentiation $y^x$ in an exponential number of steps. 
  The definition of $\m{exp}$ does not follow predicative recursion, 
  in particular since $\times$ admits no safe argument positions it
  cannot serve as stepping function. Independent on the safe mapping for $\m{exp}$, 
  rule $\rlbl{6}$ cannot be oriented using polynomial path orders.
\end{example}
\begin{example}
Finally, for a negative example consider $\RSmul$ from Example~\ref{ex:mult1} where the rule $\rlbl{4}$ is replaced by the rule
  $$\rlbl{4a}:~\times(\sn{\ms(\sn{}{x}),y}{}) \to +(\sn{\times(\sn{x,y}{})}{y})\tpkt$$
  The resulting system admits polynomial runtime complexity, but does not follow the rigid 
  scheme of predicative recursion. For this reason, it cannot be handled by \POPSTAR.\@
  Technically, 
  terms $\times(\sn{\ms(\sn{}{x}),y}{})$ and $\times(\sn{x,y}{})$ 
  is incomparable with respect to $\gsq$ independent on the precedence, and 
  consequently also orientation of left- and right-hand side with 
  \cpop{ia} fails.
\end{example}

Finally, we stress that the restriction to innermost reductions
is essential for the correctness of Theorem~\ref{t:popstar}. 
This has to do with unnecessary duplication of redexes as pointed out in Example~\ref{ex:dup}.

\begin{example}
  Reconsider the TRS $\RSdup$ from Example~\ref{ex:dup}. 
  Then $\RSdup \subseteq {\gpop}$ with any admissible precedence 
  satisfying $\m{btree} \sp \m{dup}$. Theorem~\ref{t:popstar} 
  thus implies that the (innermost) runtime complexity of $\RSdup$ is 
  polynomial. On the other hand, we already observed that $\RSdup$ admits exponentially long 
  outermost reductions.
\end{example}

\paragraph*{Proof Outline}
The proof of Theorem~\ref{t:popstar} requires a variety of ingredients. 
In Section~\ref{s:pint}, we define \emph{predicative interpretations} $\ints$
that flatten terms to \emph{sequences of terms}, essentially separating 
safe from normal arguments. This allows us to analyse terms independent from safe arguments.
In Section~\ref{s:pop} we introduce an order $\gpopv[][]$ on sequences of terms, 
that is simpler compared to $\gpop$ and does not rely on the separation
of argument positions.  
In Section~\ref{s:embed} we show that predicative interpretations
embeds innermost rewrite steps into $\gpopv[][]$:
\begin{center}
  \begin{tikzpicture}
    \newcommand{\n}{2};
    \newcommand{\dx}{1.8};
    \newcommand{\dy}{0.9};
    \newcommand{\drawterms}[2]{
      \pgfmathparse{#1}
      \node(s_#1) at (\pgfmathresult*\dx,0) {$s_{#2}$};
      \node(is_#1) at (\pgfmathresult*\dx,-\dy) {$\ints(s_{#2})$};
      \draw[->] (s_#1) to node[right] {} (is_#1); 
    }
    \newcommand{\drawrel}[1]{
      \pgfmathparse{#1}
      \node(rew_#1) at (\pgfmathresult*\dx+\dx/2,0) {$\irew$};
      \node(ord_#1) at (\pgfmathresult*\dx+\dx/2,-\dy) {$\gpopv[][]$};
    }

    \foreach \x in {1,...,\n} \drawterms{\x}{\x};

    \pgfmathparse{(\n+1)*\dx};
    \node at (\pgfmathresult,0) {$\dots$};
    \node at (\pgfmathresult,-\dy) {$\dots$};
    \drawterms{\n+2}{\ell};

    \foreach \x in {1,...,\n} \drawrel{\x};
    \drawrel{\n+1};
  \end{tikzpicture}
\end{center}
In Theorem~\ref{t:pop} we show that the length of $\gpopv[][]$ descending sequences
starting from basic terms can be bound appropriately.


\section{Predicative Interpretations}\label{s:pint}
Fix a safe mapping $\safe$ on the signature $\FS$.
In this section, we define the \emph{predicative interpretation} 
that guided by $\safe$ interpret terms as \emph{sequences}.
For this, define the \emph{normalised signature} $\FSn$ be given as
\begin{equation*}
  \FSn \defsym \bigl\{ \fn \mid f \in \FS, \normal(f) = \{i_1,\dots,i_k\} \text{ and } \ar(\fn) = k \}\bigr\}
\end{equation*}
The predicative interpretation of a term $f(\pseq{s})$ results in a sequence
$\lst{\fn(\seq[k]{a})} \append a_{k+1} \append \cdots \append a_{k+l}$, 
where $\append$ denotes concatenation of sequences and the 
sequences $a_i$ are predicative interpretations of the corresponding arguments $s_i$ ($i = 1,\dots,k+l$).
To denote sequences, we use an auxiliary variadic function symbol
$\listsym$.  Here variadic means that the arity of
$\listsym$ is finite but arbitrary.  We always write $\lseq{t}$
for $\listsym(\seq{t})$, in particular if we write $f(\seq{t})$ then $f \not=\listsym$.
Note that in the interpretations, terms have sequences as arguments.
We reflect this in the next definition.
\begin{definition}\label{d:sequencedterms}
  The set of \emph{terms with sequence arguments}
  $\TS \subseteq \TA(\FSn \uplus \set{\listsym},\VS)$
  and the set of \emph{sequences} $\LS  \subseteq \TA(\FSn \uplus \set{\listsym},\VS)$ is inductively defined as follows:
  \begin{enumerate}
  \item $\VS \subseteq \TS$, and
  \item if $\seq{t} \in \TS$ then $\lseq{t} \in \LS$, and
  \item if $\seq{a} \in \LS$ and $f \in \FSn$ then $f(\seq{a}) \in \TS$.
  \end{enumerate}
\end{definition}
We always write $a,b, \dots$, possibly extended by
subscripts, for elements from $\TS$ and $\LS$.  
The restriction of $\TS$ and $\LS$ to ground terms is denoted  
by $\GTS$ and $\GLS$ respectively.
When no confusion can arise from this we call terms with sequence arguments simply terms.
Further, we sometimes abuse set notation and write $b \in \lseq{a}$ if 
$b = a_i$ for some $i \in \{1,\dots,n\}$.
We denote by $a \append b$ the \emph{concatenation} of $a \in \TLS$ and $b \in \TLS$.
To avoid notational overhead we identify terms with singleton sequences.
Let $\tolst(a) \defsym \lst{a}$ if $a \in \TS$ and $\tolst(a) \defsym a$ if $a \in \LS$.
We set $a \append b \defsym \lst{a_1~\cdots~a_n~b_1~\cdots~b_m}$ 
where $\tolst(a) = \lseq{a}$ and $\tolst(b) = \lseq[m]{b}$.
We define the \emph{length} over $\TLS$ as $\len(a) \defsym n$ where $\tolst(a) = \lseq{a}$.
The \emph{sequence width} $\width$ (or \emph{width} for short) 
of an element $a \in \TLS$ is given recursively by
\begin{equation*}
  \width(a) \defsym
  \begin{cases}
    1 & \text{if $a$ is a variable, }\\
    \max \{1,\width(a_1),\dots,\width(a_n)\} & \text{if $a = f(\seq{a})$, and}\\
    \sum_{i=1}^n \width(a_i)
    & \text{if $a = \lseq{a}$.}
  \end{cases}
\end{equation*}
We will tacitly employ $\len(a) \leqslant \width(a)$ and $\width(a \append b) = \width(a) + \width(b)$ for all $a,b \in \TLS$.
We definite the \emph{norm} of $t \in \TERMS$ in correspondence to the depth of $t$, but 
disregard normal argument positions.
\begin{equation*}
  \norm{t} =
  \begin{cases}
    1 & \text{$t$ is a variable} \\
    1+ \max\{\norm{t_{k+1}},\dots,\norm{t_{k+l}}\} & \text{$t=f(\pseq{t})$}
  \end{cases}
\end{equation*}
Note that since all argument positions of constructors are safe, 
the norm $\norm{\cdot}$ and depth $\depth(\cdot)$ coincides on values. 
Predicative interpretations are given by two mappings $\ints$ and $\intn$:
the interpretation $\ints$ is applied on safe arguments and removes values; 
the mapping $\intn$ is applied to normal arguments and additionally encodes
the norm of the given term as tally sequence. Consequently we keep 
track of the maximal depth of values at normal argument positions.
Let $\theconst \not \in \FSn$ be a fresh constant.
To encode natural numbers $n\in \N$, 
define its \emph{tally sequence representation} $\natToSeq{n}$ 
as the sequence containing $n$ occurrences of this fresh constant:
$\natToSeq{0} = \nil$ and $\natToSeq{n+1} = \theconst \append \natToSeq{n}$.
\begin{definition}\label{d:pi}
A \emph{predicative interpretation}
is a pair $(\ints,\intn)$ of mappings $\ofdom{\ints,\intn}{\TERMS \to \TAL^{\!\ast}\!(\FS \cup \{\theconst\})}$
defined as follows:
\begin{align*}
  \ints(t) & \defsym
  \begin{cases}
    \nil & \text{ if $t$ is a value} \\
    \lst{\fn(\intn(t_1), \dots, \intn(t_k))} \append \ints(t_{k+1}) \append \cdots \append \ints(t_{k+l}) & \text{ otherwise where ($\star$)}
  \end{cases}\\
  \intn(t) & \defsym \ints(t) \append \NM{t} \tpkt
\end{align*}
Here ($\star$) stands for $t = f(\pseq{t})$.
\end{definition}

In the next section we introduce the order $\gpopv[][]$ on 
sequences $\GTLS$. 
In the subsequent section we then embed innermost $\RS$-steps into this order, 
and use $\gpopv[][]$ to estimate the length of reductions accordingly.
Since for basic terms $s = f(\pseq{u})$ in particular
$$
\ints(s) = \lst{\fn(\intn(u_1),\dots,\intn(u_k))} \append \ints(u_{k+1}) \append \cdots \append \ints(u_{k+l}) 
= \lst{\fn(\natToSeq{\depth(u_1)},\dots,\natToSeq{\depth(u_k)})}
$$
the so obtained bound will depend on depths of normal arguments only. 
To get the reader prepared for the definition of $\gpopv[][]$, 
we exemplify Definition~\ref{d:pi} on a predicative recursive TRS.\@
\begin{example}\label{ex:pint}
  Consider following predicative recursive TRS $\RS_f$ where we suppose that besides $f$, also $g$ and $h$ are defined symbols:
  \begin{align*}
    \rlbl{1}:~f(\sn{\mZ}{y}) & \to y 
    & \rlbl{2}:~f(\sn{\ms(x)}{y}) & \to g(\sn{h(\sn{x}{})}{f(\sn{x}{y})})
  \end{align*}
  Consider a substitution $\ofdom{\sigma}{\Var \to \Val}$.
  Using that $\intn(v) = \natToSeq{\depth(v)}$ for all values $v$, 
  the embedding $\ints(l\sigma) \gpopv[][] \ints(r\sigma)$ of root steps ($l \to r \in \RS_f$)
  results in the following order constraints.
  \begin{align*}
    \lst{\fn(\natToSeq{1})} & \gpopv[][] \nil && \text{from rule \rlbl{1}}\\
    \mparbox[r]{45mm}{\lst{\fn(\natToSeq{\depth(x\sigma) + 1})}} & \gpopv[][] \mparbox[l]{65mm}{\lst{\gn(\intn(h(\sn{x\sigma}{})))~\fn(\natToSeq{\depth(x\sigma)})}}  && \text{from rule \rlbl{2}}.
  \end{align*}
  where $\intn(h(\sn{x\sigma}{})) = \lst{\hn(\intn(x\sigma))}\append \NM{h(\sn{x\sigma}{})}  = \lst{\hn(\natToSeq{\depth(x\sigma)})~\theconst}$.
  Closure under context follows using standard inductive reasoning.
  To deal with steps below normal argument positions, it is also necessary to 
  orient images of $\intn$. On the TRS $\RS_f$ this results additionally in following constraints:
  \begin{align*}
    \lst{\fn(\natToSeq{1})} \append \natToSeq{\depth(y\sigma) + 1} & \gpopv[][] \natToSeq{\depth(y\sigma)} && \text{from rule \rlbl{1}}\\
    \mparbox[r]{45mm}{\lst{\fn(\natToSeq{\depth(x\sigma) + 1})}\append \natToSeq{\depth(y\sigma) + 1}} & \gpopv[][] \mparbox[l]{65mm}{\lst{\gn(\intn(h(\sn{x}{})))~\fn(\natToSeq{\depth(x\sigma)})} \append \natToSeq{\depth(y\sigma) + 1}}
    && \text{from rule \rlbl{2}}.
  \end{align*}
\end{example}
To get a polynomial bound on $\gpopv[][]$ descending sequences, we need to control
the length of right-hand sides appropriately. Precisely we will require 
that for a global constant $k \in \N$, $\len(b) \leqslant \width(a) + k$ whenever $a \gpopv[][] b$ holds.
In particular $k$ will be more than twice the maximal size of a right-hand side in the analysed TRS $\RS$. 
Note that due to the following lemma, if $l\sigma \irew r\sigma$ with $\ofdom{\sigma}{\VS \to \Val}$ is a root step of a predicative TRS $\RS$, 
then $\len(\intq(r\sigma)) \leqslant \width(\intq(l\sigma)) + k$ for $\intq \in \{\ints,\intn\}$. 
\begin{lemma}\label{l:int:len}
  Let $s = f(\pseq{s}) \in \Tb$, $t \in \TA$, $\ofdom{\sigma}{\VS \to \Val}$ and define $k \defsym 2\cdot\size{t}$. Then
  \begin{enumerate}
  \item\label{l:int:len:S} $\len(\ints(t\sigma)) \leqslant \size{t}$; and
  \item\label{l:int:len:gsq} if $s \gsq t$  then $\len(\intn(t\sigma)) \leqslant \max \{ \norm{s_1\sigma}, \dots, \norm{s_k\sigma} \} + k$; and
  \item\label{l:int:len:gpop} if $s \gpop t$ then $\len(\intn(t\sigma)) \leqslant \max \{ \norm{s_1\sigma}, \dots, \norm{s_k\sigma},\norm{s\sigma}\} + k$.
  \end{enumerate}
\end{lemma}
\begin{proof}
  The first property follows by induction on $t$, employing that $\ints(x\sigma) = \nil$.
  A standard induction on $\gpop$ (respectively $\gsq$) proves the second and third properties. 
  For the cases $s \cpop{st} t$ (respectively $s \csq{st} t$) and $s \cpop{ep} t$, we use 
  Lemma~\ref{l:gpop:val}, the remaining cases follow from induction hypothesis directly.
\end{proof}

\section{The Polynomial Path Order on Sequences}\label{s:pop}
The \emph{polynomial path order on sequences} (\emph{\POP}~for short), denoted by $\gpopv[][]$, 
constitutes a generalisation of the \emph{path order for $\FP$} as put forward in~\cite{AM05}.
Whereas we previously uses the notion of safe mapping to 
dictate predicative recursion on compatible TRSs, 
the order on sequences relies on the explicit separation of safe 
arguments as given by predicative interpretations.
Following Buchholz~\cite{B95}, we present \emph{finite approximations}
$\gpopv[k][l]$ of $\gpopv[][]$.
The parameters $k \in \N$ and $l \in \N$ are used to controls the width and depth
of right-hand sides.
Fix a precedence $\qp$ on the normalised signature $\FSn$.
We extend term equivalence with respect to $\qp$ to sequences by 
disregarding the order on elements.
\begin{definition}\label{d:eqi}
  We define $a \eqi b$ if $a = b$ or there exists a permutation $\pi$
  such that $a_i \eqi b_{\pi(i)}$ for all $i = 1,\dots,n$, 
  where either 
  (i) $a = \lseq{a}$, $b = \lseq{b}$, or 
  (ii) $a = f(\seq{a})$, $b = g(\seq{b})$ and $f \ep g$.
\end{definition}
\noindent In correspondence to $\gpop$, the order $\gpopv[k][l]$ 
is based on an auxiliary order $\gppv[k][l]$ 
defined next. The full order is then introduced in Definition~\ref{d:gpopv}. 

\begin{definition}\label{d:gppv} 
  Let $k,l \geqslant 1$.
  We define $\gppv[k][l]$ with respect to the precedence $\qp$ inductively as follows:
  \begin{enumerate}
  \item\label{d:gppv:st}
    $f(\seq{a}) \gppv[k][l] b$ if $a_i \geqppv[k][l] b$  for some $i \in \{1,\dots,n\}$;
  \item\label{d:gppv:ia}
    $f(\seq{a}) \gppv[k][l] g(\seq[m]{b})$ if $f \sp g$ and the following conditions are satisfied:
    \begin{itemize}
    \item $f(\seq{a}) \gppv[k][l-1] b_j$ for all $j = 1,\dots,m$;
    \item $m \leqslant k$;
    \end{itemize}
  \item\label{d:gppv:ialst}
    $f(\seq{a}) \gppv[k][l] \lseq[m]{b}$ if the following conditions are satisfied:
    \begin{itemize}
    \item $f(\seq{a}) \gppv[k][l-1] b_j$ for all $j = 1,\dots,m$;
    \item $m \leqslant \width(f(\seq{a})) + k$;
    \end{itemize}
  \item\label{d:gppv:ms} 
    $\lseq{a} \gppv[k][l] \lseq[m]{b}$ if the following conditions are satisfied:
    \begin{itemize}
    \item $\lseq[m]{b} \eqi c_1 \append \cdots \append c_n$;
    \item $a_i \geqppv[k][l] c_i$ for all $i = 1, \dots, n$;
    \item $a_{i_0} \gppv[k][l] c_{i_0}$ for at least one $i_0 \in \{1, \dots, n\}$;
    \item $m \leqslant \width(\lseq{a}) + k$;
    \end{itemize}
  \end{enumerate}
  Here ${\geqppv[k][l]}$ denotes ${\gppv[k][l]} \cup {\eqi}$.  We write
  $\gppv[k]$ to abbreviate $\gppv[k][k]$.
\end{definition}
\noindent Recall that the auxiliary order $\gsq$ underlying $\POPSTAR$ is used to orient normal arguments in right-hand sides. 
Similar, the auxiliary order $\gppv[k][l]$ is to orient the predicative interpretations of this
normal arguments. We exemplify the order $\gppv[k][l]$ on the Example~\ref{ex:pint}.
\begin{example}\label{ex:gppv}
  Reconsider rule $\rlbl{2}$ from Example~\ref{ex:pint} where 
  in particular $f(\sn{\ms(x)}{y}) \gsq h(\sn{x}{})$.
  Define the precedence $\fn \sp \hn \sp \theconst$. 
  First recall that by definition of the operator $\append$ we have 
  $$\natToSeq{n} = \lst{\theconst~\cdots~\theconst} = \theconst\append\cdots\append \theconst \append \nil \append \cdots \append \nil$$ 
  with $n$ occurrences of $\theconst$ and $m$ occurrences of $\nil$.
  Using $n$ times $\theconst \eqi \theconst$ and $m$ times $\theconst \cppv{ialst}[k][l] \nil$ 
  we can thus prove $\natToSeq{n+m} \cppv{ms}[k][l] \natToSeq{n}$ for all $m \geqslant 1$ whenever $l \geqslant 2$.

  Let $k \geqslant 12$ be at least twice the size of the right-hand sides, 
  and consider a substitution $\ofdom{\sigma}{\Var \to \Val}$.
  To show $\intq(f(\sn{\ms(x\sigma)}{y\sigma})) \gppv[k] \intq(h(\sn{\ms(x\sigma)}{}))$ for $\intq \in \{\ints,\intn\}$, 
  we can even show the stronger property $\fn(\natToSeq{\depth(x\sigma)} + 1) \gppv[k][10] \intq(h(\sn{\ms(x\sigma)}{}))$
  since
  \begin{align*}
    \rlbl{1}: && \natToSeq{\depth(x\sigma) + 1} 
    & \cppv{ms}[k][7] \natToSeq{\depth(x\sigma)}  
    && \text{as $\depth(x\sigma) + 1 > \depth(x\sigma)$} 
    \\
    \rlbl{2}: && \fn(\natToSeq{\depth(x\sigma)} + 1) 
    & \cppv{ia}[k][8] \hn(\natToSeq{\depth(x\sigma)})  = \ints(h(\sn{x\sigma}{}))
    && \text{using $\fn \sp \hn$ and \rlbl{1}} 
    \\
    \rlbl{3}: && \fn(\natToSeq{\depth(x\sigma)} + 1) 
    & \cppv{ialst}[k][9] \lst{\hn(\depth(x\sigma))~\theconst} = \intn(h(\sn{x\sigma}{}))
    && \text{by \rlbl{2} and $\fn(\dots) \cppv{ia}[k][8] \theconst$}
  \end{align*}
\end{example}
We arrive at the definition of the full order $\gpopv[k][l]$.
\begin{definition}\label{d:gpopv} 
  Let $k,l \geqslant 1$.
  We define $\gpopv[k][l]$ inductively as the least extension of $\gppv[k][l]$ such that:
  \begin{enumerate}
  \item\label{d:gpopv:st}
    $f(\seq{a}) \gpopv[k][l] b$ if $a_i \geqpopv[k][l] b$ for some $i \in \{1,\dots,n\}$;
  \item\label{d:gpopv:ep}
    $f(\seq{a}) \gpopv[k][l] g(\seq[m]{b})$ if $f \ep g$ 
    and following conditions are satisfied: 
    \begin{itemize}
    \item $\mset{\seq{a}}~\mextension{\gpopv[k][l]}~\mset{\seq[m]{b}}$;
    \item $m \leqslant k$;
    \end{itemize}
  \item\label{d:gpopv:ialst} 
    $f(\seq{a}) \gpopv[k][l] \lseq[m]{b}$ 
    and following conditions are satisfied: 
    \begin{itemize}
    \item $f(\seq{a}) \gpopv[k][l-1] b_{j_0}$ for at most one $j_0 \in\{1,\dots,m\}$;
    \item $f(\seq{a}) \gppv[k][l-1] b_j$ for all $j \neq j_0$;
    \item $m \leqslant \width(f(\seq{a})) + k$;
    \end{itemize}
  \item\label{d:gpopv:ms} 
    $\lseq{a} \gpopv[k][l]  \lseq[m]{b}$
    and following conditions are satisfied: 
    \begin{itemize}
    \item $\lseq[m]{b} \eqi c_1 \append \cdots \append c_n$;
    \item $a_i \geqpopv[k][l] c_i$ for all $i = 1,\dots, n$; 
    \item $a_{i_0} \gpopv[k][l] c_{i_0}$ for at least one $i_0 \in \{1,\dots, n\}$; 
    \item $m \leqslant \width(\lseq{a}) + k$;
    \end{itemize}
  \end{enumerate}
  Here ${\geqpopv[k][l]}$ denotes ${\gpopv[k][l]} \cup {\eqi}$. 
  We write $\gpopv[k]$ to abbreviate $\gpopv[k][k]$.
\end{definition}
\begin{example}\label{ex:gpopv}
  Reconsider the rules from Example~\ref{ex:pint}, 
  and let $k \geqslant 12$.
  We consider only substitutions $\ofdom{\sigma}{\VS \to \Val}$.
  First consider a rewrite step $f(\sn{\mZ}{y\sigma}) \irew y\sigma$ due to rule \rlbl{1}. Exploiting the shape of $\sigma$, we 
  have 
  $\ints(f(\sn{\mZ}{y\sigma})) = \fn(\natToSeq{1}) \cpopv{ialst}[k] \nil = \ints(y\sigma)$
  and similar 
  $$
  \intn(f(\sn{\mZ}{y\sigma})) 
  = \lst{\fn(\natToSeq{1})} \append \natToSeq{\depth(y\sigma) + 1}
  \cpopv{ms} \natToSeq{\depth(y\sigma)} 
  = \intn(y\sigma) \tpkt
  $$
  Finally consider a rewrite step $f(\sn{\ms(x\sigma)}{y\sigma}) \irew g(\sn{h(\sn{x\sigma}{})}{f(\sn{x\sigma}{y\sigma})})$ caused by rule \rlbl{2}.
  This case is slightly more involved. Essentially we use $\cpopv{ep}[k][l]$ to orient the recursive call (proof step \rlbl{5}), 
  and $\cppv{ia}[k][l]$ for the remaining elements not containing $\fn$ (proof step \rlbl{6}).
  \begin{align*}
    \rlbl{4}: && \natToSeq{\depth(x\sigma) + 1} 
    & \cpopv{ms}[k][9] \natToSeq{\depth(x\sigma)}  
    \\
    \rlbl{5}: && \fn(\natToSeq{\depth(x\sigma)} + 1)
    & \cpopv{ep}[k][9] \fn(\natToSeq{\depth(x\sigma)})
    && \text{by \rlbl{4}} 
    \\
    \rlbl{6}: && \fn(\natToSeq{\depth(x\sigma)} + 1)
    & \cppv{ia}[k][10] \gn(\intn(h(\sn{x\sigma}{})))
    && \text{using $\fn \sp \gn$ and \rlbl{3}} 
    \\
    \rlbl{7}: && \fn(\natToSeq{\depth(x\sigma)} + 1)
    & \cpopv{ialst}[k][11] \lst{\gn(\intn(h(\sn{x\sigma}{})))~\fn(\natToSeq{\depth(x\sigma)})}
    && \text{using \rlbl{5} and \rlbl{6}}
    \\
    &&& = \ints(g(\sn{h(\sn{x\sigma}{})}{f(\sn{x\sigma}{y\sigma})}))
    \\
    \rlbl{8}: && \intn(f(\sn{\ms(x\sigma)}{y\sigma})) & = 
    \lst{\fn(\natToSeq{\depth(x\sigma) + 1})}\append \natToSeq{\depth(y\sigma) + 1} \\
    &&& \cpopv{ialst}[k][k] \lst{\gn(\intn(h(\sn{x\sigma}{})))~\fn(\natToSeq{\depth(x\sigma)})}\append \natToSeq{\depth(y\sigma) + 1}
    && \text{using \rlbl{7}}
    \\
    &&& = \intn(g(\sn{h(\sn{x\sigma}{})}{f(\sn{x\sigma}{y\sigma})}))
  \end{align*}
\end{example}

The next lemma collects some frequently used properties.
\begin{lemma}\label{l:approx}
  The following properties hold for all $k \geqslant 1$ and $a,b,c_1,c_2 \in \TLS$.
  \begin{enumerate}
  \item\label{l:approx:kmon} ${\gppv[l]} \subseteq {\gpopv[l]} \subseteq {\gpopv[k]}$ for all $l \leqslant k$;
  \item\label{l:approx:modeqi} ${\eqi} \cdot {\gpopv[k]} \cdot {\eqi} \subseteq {\gpopv[k]}$;
  \item\label{l:approx:subseq} $a \gpopv[k] b$ implies ${c_1 \append a \append c_2} \gpopv[k] {c_1 \append b \append c_2}$.
  \end{enumerate}
\end{lemma}
\begin{proof}
  The first two properties follow by standard reasoning. 
  For the final property on proves $a \append c_2 \cpopv{ms}[k] b \append c_2$
  by case analysis on the assumption $a \gpopv[k] b$. Crucially, $\len(b \append c_2)$ is bounded 
  by $\width(a \append c_2) + k$ as required in $\cpopv{ms}[k]$. 
  The general property is then an easy consequence from Property~\ref{l:approx:modeqi}.
\end{proof}

Following \cite{AM05} we define $\Slow[k]$ that measures the 
$\gpopv[k]$-descending lengths on sequences. To simplify matters, 
we restrict the definition of $\Slow[k]$ to ground sequences.
As images of predicative interpretations are always ground, this suffices for our purposes.
\begin{definition}
We define $\ofdom{\Slow[k]}{\GTLS \to \N}$
as
\begin{equation*}
  \Slow[k](a) \defsym 1 + \max
 \{ \Slow[k](b) \mid b \in \GTLS \text{ and } a \gpopv[k] b
 \}\tpkt
\end{equation*}
\end{definition}
\noindent Note that due to Lemma~\eref{l:approx}{modeqi}, $\Slow[k](a) = \Slow[k](b)$ 
whenever $a \eqi b$.
Sequences are intended to act purely as a container, 
not contributing to $\Slow[k]$ themselves.
The next lemma confirms our intention, 
exploiting that conceptually clause $\cpopv{ms}[k]$ amounts to a product-wise extension of $\gpopv[k]$ to sequences.
\begin{lemma}\label{l:slowsum} 
  For $\lseq{a} \in \GLS$ it holds that $\Slow[k](\lseq{a}) = \sum_{i=1}^n \Slow[k](a_i)$.
\end{lemma}
\begin{proof}
  Let $a = \lseq{a} \in \GLS$.
  We first show $\Slow[k](a) \geqslant \sum_{i=1}^n \Slow[k](a_i)$.
  Let $b,c \in \GTLS$ and consider maximal sequences
  $b \gpopv[k] b_1 \gpopv[k] \cdots \gpopv[k] b_o$ and
  $c \gpopv[k] c_1 \gpopv[k] \cdots \gpopv[k] c_p$.
  Using Lemma~\eref{l:approx}{subseq} repeatedly we get
  $b \append c \gpopv[k] b_1 \append c \gpopv[k] \cdots \gpopv[k] b_o \append c$, 
  similar
  $c \append b_o \gpopv[k] c_1 \append b_o \gpopv[k] \cdots \gpopv[k] c_p \append b_o$.
  Since $b_o \append c \eqi c \append b_o$ and employing
  Lemma~\eref{l:approx}{subseq} we see 
  $\Slow[k](b \append c) \geqslant \Slow[k](b) + \Slow[k](c)$
  for all $b,c \in \GTLS$. We conclude
  $
  \Slow[k](a) =  \Slow[k](a_1 \append \cdots \append a_n) \geqslant \sum_{i=1}^n \Slow[k](a_i)
  $ 
  with a straight forward induction on $n$.

  It remains to verify $\Slow[k](a) \leqslant \sum_{i=1}^n \Slow[k](a_i)$.
  For this we show that $a \gpopv[k] b$ implies $\Slow[k](b) < \sum_{i=1}^n \Slow[k](a_i)$
  by induction on $\Slow[k](a)$.
  Consider the base case $\Slow[k](a) = 0$.
  Since $a$ is ground it follows that $a = \nil$, the claim is trivially satisfied.
  For the inductive step $\Slow[k](a) > 1$, 
  let $a \gpopv[k] b$.
  Since $a$ is a sequence, $a \cpopv{ms} b$.
  Hence $b \eqi b_1 \append \cdots \append b_n$ where $a_i \geqpopv[k] b_i$ 
  and thus $\Slow[k](b_i) \leqslant \Slow[k](a_i)$ for all $i =1,\dots,n$.
  Additionally $a_{i_0} \gpopv[k] b_{i_0}$ and hence $\Slow[k](b_{i_0}) < \Slow[k](a_{i_0})$ 
  for at least one $i_0 \in \{1,\dots,n\}$.
  As in the first half of the proof, one verifies $\Slow[k](b_i) \leqslant \Slow[k](b)$
  for all $i = 1,\dots,n$.
  Note $\Slow[k](b) < \Slow[k](a)$ as $a \gpopv[k] b$, 
  hence induction hypothesis is applicable to $b$ and all $b_i$ ($i = 1,\dots,n$).
  It follows that 
  \begin{align*}
    \Slow[k](b) = \sum_{c \in b} \Slow[k](c) 
    = \sum_{i=1}^n \sum_{c \in b_i} c 
    = \sum_{i=1}^n \Slow[k](b_i) 
    < \sum_{i=1}^n \Slow[k](a_i) \tpkt
  \end{align*}
  This concludes the second half of the proof.
\end{proof}

The central theorem of this section states that $\Slow[k](f(a_1, \dots, a_n))$
is polynomial in $\sum_{i}^n \Slow[k](a_i)$, 
where the polynomial bound depends only on $k$ and the rank $p$ of $f$.
The proof of this is rather involved.
To cope with the multiset comparison underlying $\cpopv{ep}[k]$,
we introduce as a first step an \emph{order-preserving} extension $\MSlow{n}{k}$ of $\Slow[k]$ to multisets of sequences, 
in the sense that $\MSlow{n}{k}(\seq{a}) > \MSlow{m}{k}(\seq[m]{b})$ holds 
whenever $\mset{\seq[n]{a}} \mextension{\gpopv[k]} \mset{\seq[m]{b}}$ (provided $k \geqslant m,n$, c.f. Lemma~\ref{d:mslow}).
As the next step toward our goal, we estimate $\Slow[k](f(a_1,\dots,a_n))$ in terms 
of $\MSlow{n}{k}(\seq{a})$
whenever $n \leqslant k$ and $\rk(f) \leqslant k$.
Technically we bind following functions by polynomials $q_{k,p}$. 
\begin{definition}
  For all $k,p \in \N$ with $k \geqslant 1$
  we define $\ofdom{\Fpop{k}{p}}{\N \to \N}$ as 
  \begin{multline*}
    \Fpop{k}{p}(m) \defsym \max\{ \Slow(f(\seq{a})) \mid \\
    f(\seq{a}) \in \GTS,
    ~\rk(f) \leqslant p,~n \leqslant k
    \text{ and } \MSlow{n}{k}(a_1,\dots,a_n) \leqslant m \} \tpkt
  \end{multline*}
\end{definition}
Noting that also $\MSlow{n}{k}(\seq{a})$ is polynomial in $\max_{i=1}^n\Slow[k](a_i)$, say $q_k$, which depends only on $k$,  
we obtain $\Slow(f(\seq{a})) \leqslant q_{k,p}(q_{k}(\max_{i=1}^n \Slow[k](a_i)))$ whenever $k \geqslant n$.

\medskip

The definition of $\MSlow{n}{k}$ is defined in terms of an order-preserving homomorphism from $\msetover(\N)$ to $\N$. 
To illustrate the construction carried out below, consider the following example.
\begin{example}
  Let $k \geqslant 1$ and 
  let $c > m_1 \geqslant \dots \geqslant m_k$ be natural numbers in descending order, dominated by $c \in \N$.
  Consider multisets $\msetover(\N)$ of size $k$.
  If we conceive such multisets as base-$c$ representations of numbers using $k$ digits, then we can form
  a chain $M_1 \mextension{>} M_2 \mextension{>} \dots$ that can be understood as a decreasing counter
  that wrongly wraps from $\{m_1,\dots, m_i+1,0,\dots,0\}$ to $\{m_1,\dots, m_i,m_i,\dots,m_i\}$.
  It is not difficult to prove that the maximal length of such a chain is bounded by
  \begin{align*}
  \sum_{m_2=1}^{m_1} \dots \sum_{m_{k-1}=1}^{m_{k-2}} \sum_{m_k=1}^{m_{k-1}} m_k 
  ~\in~\sum_{m_2=1}^{m_1} \dots \sum_{m_{k-1}=1}^{m_{k-2}} \Omega( m_{k-1}^2 )
  ~\in~\Omega( m_1^k ) \tpkt    
  \end{align*}
\end{example}
We now show that this upper bound serves also as a lower bound for 
multisets $\msetover(\N)$ of size $n \leqslant k$.
As in the example,
the function $\ofdom{\homo{n}{k}{c}}{\N^l \to \N}$ (where $n \leqslant k$) defined below interprets  
multisets $M \in \msetover(\N)$ as natural numbers encoded in base-$c$ with $k$ digits,
where the \nth{$i$} largest $m_i \in M$ represents the \nth{$i$} most significant digit.
Formally, for $k\geqslant n \in \N$ and $c \in \N$ we define the family of functions $\ofdom{\homo{n}{k}{c}}{\N^l \to \N}$ 
such that 
\begin{align*}
  \homo{n}{k}{c}(\seq[n]{m}) = \sum_{i=1}^n \sort{n}(\seq[n]{m},i) \cdot c^{(k - i)} \tpkt
\end{align*}
Here $\sort{n}(\seq[n]{m},i)$ denote the \nth{$i$}\ element of $\seq[n]{m}$ sorted in descending order, i.e., 
$\sort{n}(\seq[m]{n},i) \defsym m_{\pi(i)}$ 
for $i =1,\dots,n$ and some permutation $\pi$ such that $m_{\pi(i)} \geqslant m_{\pi(i+1)}$ ($i \in \{1,\dots,n-1\}$).
\begin{lemma}\label{l:homo}
  Let $k,n,c \in \N$ such that $k \geqslant 1$ and $k \geqslant n$. Then for all $\seq[l]{n} \in \N$ we have:
  \begin{enumerate}
  \item\label{l:homo:4} $c > \max\set{\seq[l]{n}}$ implies $c^{n} > \homo{l}{n}{c}(\seq[n]{m})$.
  \item\label{l:homo:5} $\mset{\seq[n]{m}} \mextension{>} \mset{\seq[n']{m'}}$ implies $\homo{n}{k}{c}(\seq[n]{m}) > \homo{n'}{k}{c}(\seq[n']{m'})$
    for all $c > \seq[n]{m} \geqslant 1$.
  \end{enumerate}
\end{lemma}
The mapping $\MSlow{n}{k}$ is obtained by extend $\homo{l}{k}{\cdot}$ to multisets over $\GTLS$.
\begin{definition}\label{d:mslow}
  Let $k,n \in \N$ such that $k \geqslant n$.
  We define $\ofdom{\MSlow{n}{k}}{\GLS^n \to \N}$ as 
  $$
  \MSlow{n}{k}(\seq[n]{a}) \defsym \homo{n}{k}{c}(\Slow[k](a_1),\dots,\Slow[k](a_n))
  $$
  where $c = 1 + \max \set{\Slow[k](a_i) \mid i \in \{1,\dots,n\}}$.
\end{definition}
\noindent By Lemma~\eref{l:homo}{4}, $\MSlow{l}{k}(\seq[l]{a})$ is polynomially bounded 
in $\Slow[k](a_i)$ ($i = 1,\dots,l$).
By Lemma~\eref{l:homo}{5} we obtain that 
$\MSlow{l}{k}$ is indeed order preserving as outlined above.
\begin{lemma}\label{l:slowpoly}
  Let $\seq{a},\seq[m]{b}\in \GTLS$ and let $k \geqslant m,n$. Then
  $$
  \mset{\seq{a}} \mextension{\gpopv[k]} \mset{\seq[m]{b}} 
  \quad \IImp \quad \MSlow{n}{k}(\seq{a}) > \MSlow{m}{k}(\seq[m]{b}) \tpkt
  $$
\end{lemma}

In Theorem~\ref{t:pop} below we prove $\Fpop{k}{p}(m) \leqslant c_{k,p} \cdot {(m+2)}^{d_{k,p}}$,
where the constants $c_{k,p},d_{k,p} \in \N$ are defined as follows:
$d_{k,0} \defsym k+1$ and $d_{k,p+1} \defsym {(d_{k,p})}^{k}+1$;
further we set $c_{k,0} \defsym k^k$ and $c_{k,p+1} \defsym {(c_{k,p} \cdot k)}^e$ 
where $e = {\sum_{i=0}^k {(k \cdot d_{k,p})}^i}$. 
Inevitably the proof of Theorem~\ref{t:pop} is technical, 
the reader is advised to skip the formal proof on the first read.
Theorem~\ref{t:pop} is proven by induction on $p$ and $m$.
Consider term $f(\seq{a})$ with $k \geqslant n$ and $\MSlow{n}{k}(\seq{a}) \leqslant m$.
At the heart of the proof, we have to show that 
$c_{k,p} \cdot {(m+2)}^{d_{k,p}} > \Slow[k](b)$ for arbitrary $b$ with $f(\seq{a}) \gpopv[k] b$. 
The most involved case is $f(\seq{a}) \cpopv{ialst}[k] b$ for $b = \lseq[o]{b}$.
Here it is fundamental to give precise bounds on the elements $b_j$ 
with $f(\seq{a}) \gppv[k][l] b_j$.
Since $\gppv[k][l]$ constraints $b_j$ to only contain symbols ranked below $\rk(f) = p$ in the precedence, 
conceptually $\Slow[k](b_j)$ is bounded by iterated application of the induction hypothesis on $p$.
Since $l$ essentially controls the depth of $b_j$ (compare Example~\ref{ex:gpopv}), 
$l$ serves as a bound on the number of iterations.
To properly account for all cases of $\gppv[k][l]$, matters get slightly more involved.
To bind $\Slow[k](b_j)$ sufficiently, 
we define for $l \in \N$ a family of auxiliary functions 
$\ofdom{g_{l,k,p}}{\N \to \N}$ such that 
\begin{align*}
  g_{l,k,p}(m) & = 
  \begin{cases}
    k^l \cdot m^l & \text {if $p = 0$ or $l = 0$, and} \\
    c_{k,p-1} \cdot (m \cdot g_{l-1,k,p}(m))^{k \cdot d_{k,p-1}} & \text{otherwise.}
  \end{cases}
\end{align*}
Having as premise the induction hypothesis (on $p$) of the main proof, the next lemma verifies that 
$g_{l,k,\rk(f)}(m+2)$ sufficiently binds $\gppv[k][l]$-descendants of $f(\seq{a})$.
\begin{lemma}\label{l:pop:aux}
  Let $f(\seq{a})  \in \GTS$.
  Let $k \geqslant n$ and 
  $m \geqslant \MSlow{n}{k}(\seq{a})$.
  Suppose $\Fpop{k}{p}(m') \leqslant c_{k,p} {(m'+2)}^{d_{k,p}}$ for all $p < \rk(f)$ and $m'$.
  Then $f(\seq{a}) \gppv[k][l] b$ implies $\Slow(b) \leqslant g_{k,l,\rk(f)}(m+2)$
  for all $b \in \GTLS$.
\end{lemma}
\begin{proof}
  We prove the claim by induction on $l$
  and case analysis on $f(\seq{a}) \gppv[k][l] b$.
  First note that $f(\seq{a})\cppv{st}[k][l] b$ implies 
  that $a_i \gppv[k][l] b$ for some $i \in \{1,\dots,n\}$
  and consequently $\Slow(b) \leqslant \Slow(a_i)$. 
  As by definition $\Slow(a_i) \leqslant m$ the lemma follows trivially.

  As in the base case $l = 1$ either $b = \nil$ or $f(\seq{a}) \cppv{st}[k][l] b$, 
  it suffices to consider only the remaining cases of the inductive step.
  Assuming $f(\seq{a}) \gppv[k][l+1] b$ we show 
  $\Slow(b) \leqslant g_{k,l+1,\rk(f)}(m+2)$.
  \begin{description}[leftmargin=0.3cm]
  \item[\dcase{$f(\seq{a}) \cppv{ia}[k][l+1] b$ where $b = g(\seq[o]{b})}$] 
    Then 
    $f(\seq{a}) \gppv[k][l] b_j$ for all $j = 1,\dots,o$, and $f \sp g$.
    Set $m' \defsym \MSlow{o}{k}(\seq[o]{b})$. We have
    \begin{align*}
      m' & < \max\set{\Slow[k](b_j) + 1 \mid j \in \{1,\dots,o\}}^k 
      && \text{by definition and Lemma~\eref{l:homo}{4}} \\
      & \leqslant {(g_{k,l,\rk(f)}(m+2)+ 1)}^k
      && \text{applying induction hypothesis $o$ times}
    \end{align*}
    As the assumption also gives $\rk(g) < \rk(f)$ 
    we have
    \begin{align*}
      \Slow(b) 
      & \leqslant \Fpop{k}{\rk(g)}(m') && \text{by definition of $\Fpop{k}{\rk(g)}$} \\
      & \leqslant c_{k,\rk(g)} \cdot {(m' + 2)}^{d_{k,\rk(g)}} && \text{by assumption} \\
      & \leqslant c_{k,\rk(f)-1} \cdot {(m' + 2)}^{d_{k,\rk(f)-1}} && \text{as $\rk(g) < \rk(f)$} \\
      & < c_{k,\rk(f)-1} \cdot {({(g_{k,l,\rk(f)}(m+2)+ 1)}^k + 2)}^{d_{k,\rk(f)-1}} && \text{substituting bound for $m'$} \\
      & \leqslant c_{k,\rk(f)-1} \cdot {(g_{k,l,\rk(f)}(m+2)+ 3)}^{k \cdot d_{k,\rk(f)-1}} && \text{using $1 \leqslant k$}\\
      & \leqslant c_{k,\rk(f)-1} \cdot {((m+2) \cdot g_{k,l,\rk(f)}(m+2))}^{k \cdot d_{k,\rk(f)-1}} && \text{using $2 \leqslant g_{k,l,\rk(f)}(m+2)$}\\
      & = g_{k,l+1,\rk(f)}(m+2) && \text{using $\rk(f) > 0$}\tpkt
    \end{align*}
  \item[\dcase{$f(\seq{a}) \cppv{ialst}[k][l+1] b$ where $b = \lseq[o]{b}}$]
    Ordering constraints give $o \leqslant \width(a) + k$
    and $f(\seq{a}) \gppv[k][l] b_j$ ($j = 1,\dots,o$).
    Exploiting that $a_i$ is ground, a standard induction shows 
    that $\width(a_i) \leqslant \Slow[k](a_i)$, and consequently 
    $\width(a_i) \leqslant m$.
    Thus
    \begin{align}
      \label{e:bindwidth}
      \tag{\dag}
      o \leqslant \width(a) + k 
      = \max\set{1,\width(a_1), \dots,\width(a_n)} + k
      \leqslant m + k \leqslant k \cdot (m + 1) \tpkt
    \end{align}
    If $\rk(f) = 0$ then 
    we see
    \begin{align*}
      \Slow[k](b) & = \sum_{j=1}^o \Slow[k](b_i) && \text{using Lemma~\ref{l:slowsum}}\\
      & \leqslant \sum_{j=1}^o g_{k,l0}(m+2) && \text{applying induction hypothesis $o$ times}\\
      & \leqslant k \cdot (m + 1) \cdot g_{k,l0}(m+2) && \text{using \eqref{e:bindwidth}} \\
      & = k \cdot (m + 1) \cdot k^{l} \cdot {(m + 2)}^{l} && \text{by assumption $\rk(f) = 0$} \\
      & < k^{l+1} \cdot {(m + 2)}^{l+1}
      = g_{k,l+10}(m+2) \tpkt
    \end{align*}
    Otherwise $\rk(f) > 0$ and we conclude
    \begin{align*}
      \Slow[k](b) & \leqslant k \cdot (m + 1) \cdot g_{k,l,\rk(f)}(m+2)
      && \text{as in the case $\rk(f) = 0$}\\
      & < c_{k,\rk(f)-1} \cdot {((m + 2) \cdot g_{k,l,\rk(f)}(m+2))}^{k \cdot d_{k,\rk(f)-1}} 
      && \text{as $k \leqslant c_{k,\rk(f)-1}$ and $1 < k \cdot d_{k,\rk(f)-1}$} \\
      & = g_{k,l+1,\rk(f)}(m+2) && \text{by assumption $\rk(f) > 0$}\tpkt
    \end{align*}
  \end{description}
\end{proof}

\begin{theorem}\label{t:pop}
  For all $k,p \in \N$ there exist constants $c,d \in \N$ (depending only on $k$ and 
  $p$) such that for all $m$: $\Fpop{k}{p}(m) \leqslant c \cdot {(m+2)}^{d}$.
\end{theorem}
\begin{proof}
  Fix $a = f(\seq[n]{a}) \in \GTS$ such that $\rk(f) = p$, $k \geqslant n$ and $\MSlow{n}{k}(a_1,\dots,a_n) \leqslant m$. 
  To show the theorem, we prove that for all $b$ with $a \gpopv[k] b$ 
  we have $\Slow[k](b) < c_{k,p} \cdot {(m+2)}^{d_{k,p}}$ by induction on the rank $p$ and side induction on $m$.

  \smallskip

  \noindent\basecase{$ p = 0 $}
  The base case of the side induction is trivial, so consider the inductive step $m > 0$.
  We first prove $\Slow[k](b) < k^k \cdot {(m+1)}^{k+1} + k^l \cdot {(m + 2)}^l$
  by induction on $\gpopv[k][l]$.
  \begin{description}[leftmargin=0.3cm]
  \item[\dcase{$f(\seq{a}) \cpopv{st}[k][l] b}$] Then $a_i \geqpopv[k][l] b$, and 
    we conclude since $\Slow[k](b) \leqslant \Slow[k](a_i) \leqslant m$ using 
    the assumptions and Lemma~\eref{l:homo}{4}.

  \item[\dcase{$f(\seq{a}) \cpopv{ep}[k][l] b$ where $g(\seq[o]{b})$}] 
    The ordering constraints give $o \leqslant k$, $f \ep g$ and
    $\mset{\seq{a}} \mextension{\gpopv[k][l]} \mset{\seq[o]{b}}$.
    Set $m' \defsym \MSlow{o}{k}(\seq[o]{b})$. 
    Hence $m' < \MSlow{n}{k}(\seq[n]{a}) \leqslant m$
    by Lemma~\ref{l:slowpoly} and assumption $n \leqslant k$.
    Thus side induction hypothesis gives
    $\Fpop{k}{0}(m') \leqslant c_{k,0} \cdot {(m' + 2)}^{c_{k,0}} = k^k {(m' + 2)}^{k+1}$.
    As the ordering constraints imply $\rk(g) = \rk(f) = 0$
    we conclude
    \begin{align*}
      \Slow[k](g(\seq[o]{b})) 
      & \leqslant \Fpop{k}{0}(m') 
      && \text{by definition of $\Fpop{k}{0}$} \\
      & = c_{k,0} \cdot {(m' + 2)}^{d_{k,0}}
      && \text{by side induction hypothesis}\\
      & = k^k \cdot {(m' + 2)}^{k+1}
      && \text{by definition}\\
      & < k^k \cdot {(m+1)}^{k+1} + k^{l} \cdot {(m + 2)}^{l}
      && \text{using $m' < m$.}
    \end{align*}

  \item[\dcase{$f(\seq{a}) \cpopv{ialst}[k][l]$ where $\lseq[o]{b}$}] 
    The ordering constraints give
    (i) $a \gpopv[k][l-1] b_{j_0}$ for some $j_0 \in\{1,\dots,o\}$,
    (ii) $a \gppv[k][l-1] b_j$ for all $j \neq j_0$, and
    (iii) $o \leqslant \width(a) + k$.
    By induction hypothesis on (i) we get $\Slow[k](b_{j_0}) < k^k \cdot {(m+1)}^{k+1} + k^{l-1} \cdot {(m + 2)}^{l-1}$, 
    the preparatory Lemma~\ref{l:pop:aux} on (ii) gives $\Slow[k](b_j) \leqslant k^{l-1} \cdot {(m + 2)}^{l-1}$ for $j \not = j_o$.
    Exactly as in Equation \eqref{e:bindwidth} we obtain $o \leqslant k \cdot (m + 1) < k \cdot (m + 2)$ from (iii).
    Summing up we have
    \begin{align*}
      \Slow[k](b) & = \sum_{j =1}^o \Slow[k](b_j) 
      && \text{by Lemma~\ref{l:slowsum}}\\
      & < k^k \cdot {(m+1)}^{k+1} + k^{l-1} \cdot {(m + 2)}^{l-1} 
      && \text{by induction hypothesis on (i), and}\\
      & \quad + (k \cdot (m+2) - 1) \cdot k^{l-1} \cdot {(m + 2)}^{l-1}  
      && \text{using $o < k \cdot (m+2)$ and Lemma~\ref{l:pop:aux} on (ii)}\\
      & =  k^k \cdot {(m+1)}^{k+1} + k^{l} \cdot {(m + 2)}^{l} \tpkt
    \end{align*}

  \end{description}
  
  Since $\gpopv[k] = \gpopv[k][k]$ this preparatory step gives 
  \begin{align*}
    \Slow[k](b) < k^k \cdot {(m+1)}^{k+1} + k^k \cdot {(m + 2)}^k 
    \leqslant k^k \cdot {(m + 2)}^{k+1}
  \end{align*}
  and concludes the base case.

  \smallskip

  \noindent\indcase{}
  By induction hypothesis on $p$ we get 
  $\Fpop{k}{p}(m) \leqslant c_{k,p} \cdot {(m+2)}^{d_{k,p}}$, side induction hypothesis gives 
  $\Fpop{k}{p+1}(m') \leqslant c_{k,p+1} \cdot {(m+2)}^{d_{k,p+1}}$ for all $m' < m$. 
  A standard induction reveals $g_{l,k,p+1}(n) \leqslant c_{k,p}^{\sum_{i=0}^{l-1} {(k \cdot d_{k,p})}^i} \cdot n^{\sum_{i=1}^{l-1} {(k \cdot d_{k,p})}^i}$ 
  for all $n \in \N$.
  We continue with the proof of the lemma, and show that for all $l \geqslant 1$, if  $f(\seq{a}) \gpopv[k][l] b$ then
  \begin{equation}
    \label{t:pop:a}
    \tag{$\ddag$}
    \Slow[k](b) \leqslant c_{k,p+1}  \cdot {(m+1)}^{d_{k,p+1}} + c_{k,p+1} \cdot {(m + 2)}^{{(k \cdot d_{k,p})}^l}
  \end{equation}
  by induction on $l$.
  The only interesting case is $a \cpopv{ialst}[k][l+1] b$.
  Then $b = \lseq[o]{b}$ with 
  (i) $a \gpopv[k][l] b_{j_0}$ for some $j_0 \in \{1,\dots,o\}$,
  (ii) $a \gppv[k][l] b_j$ for all $j \neq j_0$, and
  (iii) $o \leqslant \width(a) + k$.
  By induction hypothesis on (i) we get $\Slow[k](b_{j_0}) \leqslant c_{k,p+1} \cdot {(m+1)}^{d_{k,p+1}} + c_{k,p+1} \cdot {(m + 2)}^{{(k \cdot d_{k,p})}^{l}}$, 
  Lemma~\ref{l:pop:aux} on (ii) gives $\Slow[k](b_j) \leqslant g_{l,k,p+1}(m + 2)$ for $j \not = j_o$ and 
  (iii) gives $o \leqslant k \cdot (m + 1)$ as in Equation~\eqref{e:bindwidth}.
  Summing up we see
  \begin{align*}
    \Slow[k](b) & = \sum_{j=1}^{o} \Slow[k](b_j) && \text{by Lemma~\ref{l:slowsum}}\\
    & \leqslant c_{k,p+1} \cdot {(m+1)}^{d_{k,p+1}} + c_{k,p+1} \cdot {(m + 2)}^{{(k \cdot d_{k,p})}^l} && \text{by induction hypothesis}\\
    & \quad + k \cdot (m+1) \cdot g_{l,k,p+1}(m+2) && \text{by Lemma~\ref{l:pop:aux} and bound on $o$}\\
    & \leqslant c_{k,p+1} \cdot {(m+1)}^{d_{k,p+1}} + c_{k,p+1} \cdot {(m + 2)}^{{(k \cdot d_{k,p})}^{l}} \\
    & \quad + k \cdot (m+1) \cdot c^{\sum_{i=0}^{l-1} {(k \cdot d_{k,p})}^i} \cdot {(m+2)}^{\sum_{i=1}^{l-1} {(k \cdot d_{k,p})}^i} &&  \text{bound on $g_{l,k,p+1}(m+2)$}\\
    & < c_{k,p+1} \cdot {(m+1)}^{d_{k,p+1}} + c_{k,p+1} \cdot {(m + 2)}^{{(k \cdot d_{k,p})}^{l}} \\
    & \quad + c_{k,p+1} \cdot {(m+2)}^{\sum_{i=0}^{l-1} {(k \cdot d_{k,p})}^i}
    && \text{using $k \cdot c_{k,p}^{\sum_{i=0}^{l-1} {(k \cdot d_{k,p})}^i} \leqslant c_{k,p+1}$}\\
    & \leqslant c_{k,p+1} \cdot {(m+1)}^{d_{k,p+1}} + c_{k,p+1} \cdot {(m+2)}^{\sum_{i=0}^{l} {(k \cdot d_{k,p})}^i} \\
    & \leqslant c_{k,p+1} \cdot {(m+1)}^{d_{k,p+1}} + c_{k,p+1} \cdot {(m+2)}^{{(k \cdot d_{k,p})}^{l+1}} 
  \end{align*}
  as desired.
  Using $\eqref{t:pop:a}$, $\gpopv[k] = \gpopv[k][k]$ 
  and ${(k \cdot d_{k,p})}^k < {(k \cdot d_{k,p})}^k + 1 < d_{k,p+1}$ we finally get
  \begin{align*}
    \Slow[k](b) & \leqslant c_{k,p+1}  \cdot {(m+1)}^{d_{k,p+1}} + c_{k,p+1} \cdot {(m + 2)}^{{(k \cdot d_{k,p})}^k} \\
    & = c_{k,p+1} \cdot ({(m+1)}^{d_{k,p+1}} + {(m + 2)}^{{(k \cdot d_{k,p})}^k}) \\
    & < c_{k,p+1} \cdot {(m + 2)}^{d_{k,p+1}}
  \end{align*}
  and conclude the inductive case.
\end{proof}

As a consequence, the number of $\gpopv[k]$-descents on 
basic terms interpreted with predicative interpretation $\ints$ 
is polynomial in sum of depths of normal arguments.

\begin{corollary}\label{c:pop}
  Let $f \in \DS$ with at most $k$ normal arguments. There exists a constant $d \in \N$ depending only on $k$ such that:
  $$\Slow[k](\ints(f(\sn{\seq{m}{u}}{\vec{v}}))) \in \bigO\bigl({(\max_{i=1}^m \depth(u_i))}^{d}\bigr)$$
  for all $\seq[m+n]{u} \in \Val$.
\end{corollary}
\begin{proof}
  Let $s = f(\sn{\seq{m}{u}}{\vec{v}})$ be as given by the corollary. Recall that
  \begin{align*}
    \ints(s) 
    & = \lst{\fn(\intn(u_1), \dots, \intn(t_{u_m}))} \append \ints(u_{m+1}) \append \cdots \append \ints(u_{m+n}) \\
    & = \lst{\fn(\natToSeq{\depth(u_1)}, \dots, \natToSeq{\depth(u_m)})}
  \end{align*}
  As $\Slow[k](\theconst)$ is constant, say $\Slow[k](\theconst) = c$, by Lemma~\ref{l:slowsum} we see that
  $\Slow[k](\natToSeq{\depth(u_i)}) = c \cdot \depth(u_i)$. 
  Putting things together is tedious but not difficult:
  \begin{align*}
    \Slow[k](s) =
    & \Slow[k](\fn(\natToSeq{\depth(u_1)}, \dots, \natToSeq{\depth(u_m)})) 
    && \text{by Lemma~\ref{l:slowsum}} 
    \\
    & \leqslant \Fpop{k}{\rk(f)}\bigl(\MSlow{l}{k}(\natToSeq{\depth(u_1)}, \dots, \natToSeq{\depth(u_m)})\bigr) 
    \\
    & \leqslant \Fpop{k}{\rk(f)}\Bigl({\bigl(1+ \max_{i=1}^m \Slow[k](\NM{u_i})\bigr)}^k\Bigr) 
    && \text{by Lemma~\eref{l:homo}{4}} 
    \\
    & \leqslant \Fpop{k}{\rk(f)}\Bigl({\bigl(c \cdot (1 + \max_{i=1}^m \depth(u_i))\bigr)}^k\Bigr)
    && \text{using $\Slow[k](\NM{u_i}) \leqslant c \cdot \depth(u_i)$}
    \\
    & \in \bigO\Bigl({\bigl(c \cdot (1 + \max_{i=1}^m \depth(u_i))\bigr)}^{k+d'}\Bigr)
    && \text{by Theorem~\ref{t:pop}} 
    \\
    & \in \bigO\bigl({\max_{i=1}^m \depth(u_i)}^{k+d'}\bigr)
  \end{align*}
  Set $d \defsym k + d'$ and note that $d$ depends only on $k$ and $\rk(f)$ as desired.
\end{proof}


\section{Predicative Embedding}\label{s:embed}

Fix a predicative recursive TRS $\RS$ and signature $\FS$, 
and let $\gpop$ be the polynomial path order underlying $\RS$ 
based on the (admissible) precedence $\qp$.
We denote by $\qp$ also the homomorphic precedence on $\FSn$ 
given by: $\fn \ep \gn$ if $f \ep g$ and $\fn \sp \gn$ if $f \sp g$.
Further, we set $f \sp \theconst$ for all $\fn \in \FSn$.
We denote by $\gpopv[\ell]$ (and respectively $\gppv[\ell]$) the approximation 
given in Definition~\ref{d:gpopv} (respectively Definition~\ref{d:gppv}) with underlying precedence $\qp$.
We now establish the embedding of $\irew$ into $\gpopv[\ell]$
for some $\ell$ depending only on $\RS$.
To simplify matters, we suppose for now that $\RS$ is \emph{completely defined}.
Since then normal forms and values coincide, $s \irew t$ 
if $s = C[l\sigma]$ and $t = C[r\sigma]$ where 
${l \to r} \in \RS$ and all arguments of $l\sigma$ are values.
In particular, this implies that the substitution $\sigma$
maps variables to values.

Lemma~\ref{l:embed:root} below proves the embedding of root steps 
for the case $l \gpop r$. In Lemma~\ref{l:embed:ctx} we then show that the embedding is closed under contexts.
The next lemma, exploited in Lemma~\ref{l:embed:root}, 
connects the auxiliary orders $\gsq$ and $\gppv[k][l]$ (compare Example \ref{ex:gppv}).
\begin{lemma}\label{l:embedgsq:root}
  Suppose $s = f(\pseq{s}) \in \BASICS$, $t \in \TERMS$ and $\ofdom{\sigma}{\VS \to \Val}$. 
  Then for predicative interpretation $\intq \in \{\ints, \intn\}$
  we have 
  $$
   s \gsq t \quad \IImp \quad \fn(\intn(s_1\sigma), \dots, \intn(s_k\sigma)) \gppv[2\cdot\size{t}] \intq(t\sigma) \tpkt
  $$
\end{lemma}
\begin{proof}
  Let $s = f(\pseq{s}) \in \BASICS$, $t \in \TERMS$. 
  We continue by induction on the definition of $\gsq$. 
  \begin{description}[leftmargin=0.3cm]
  \item[\dcase{$s \csq{st} t$}]
    Then $s_i \geqsq t$ for some normal argument position $i \in \{1,\dots,k\}$.

    Note that by assumption $s \in \BASICS$, $s_i \in \Val$
    and so Lemma~\ref{l:gpop:val} (employing ${\gsq} \subseteq {\gpop}$)
    gives $s_i \esuperterm t$ and $t \in \Val$, consequently 
    $t\sigma \in \Val$
    and furthermore $\norm{s_i\sigma} \geqslant \norm{t\sigma}$.
    As $t\sigma \in \Val$, we get
    $\fn(\intn(s_1\sigma), \dots, \intn(s_k\sigma)) \cppv{ialst}[2 \cdot \size{t}] \nil = \ints(t\sigma)$ which concludes the case $\intq = \ints$.
    For the case $\intq = \intn$, 
    observe $\intn(s_i\sigma) = \NM{s_i\sigma}$ and $\intn(t\sigma) = \NM{t\sigma}$
    since both $s_i\sigma$ and $t\sigma$ are values.
    If $\norm{s_i\sigma} = \norm{t\sigma}$ then obviously 
    $\intn(s_i\sigma) = \intn(t\sigma)$. Otherwise 
    $\norm{s_i\sigma} > \norm{t\sigma}$ and then
    $\intn(s_i\sigma) \cppv{ms}[2\cdot\size{t}] \intn(t\sigma)$, 
    employing $\theconst \cppv{ialst}[2\cdot\size{t}] \nil$.
    Hence overall $\intn(s_i\sigma) \geqppv[2\cdot\size{t}] \intn(t\sigma)$.
    Since the position $i$ is normal, $\intn(s_i\sigma)$ is a direct subterm of $\fn(\intn(s_1\sigma), \dots, \intn(s_k\sigma))$
    and we conclude $\fn(\intn(s_1\sigma), \dots, \intn(s_k\sigma)) \cppv{st}[2\cdot\size{t}] \intn(t\sigma)$ as desired.

  \item[\dcase{$s \csq{ia} t$}]
    By the assumption $t = g(\pseq[m][n]{t})$ where $f \sp g$ and $s \gsq t_i$ for all $j \in =1,\dots,m+n$.
    
    We consider the more involved case $t\not\in\Val$. 
    Let $\ints(t_i\sigma) = \lst{v_{i,1}~\cdots~v_{i,j_i}}$ for all safe argument positions 
    $i = m+1, \dots m+n$ of $g$, i.e., 
    \begin{align*}
    \ints(t\sigma) 
    & = \lst{\gn(\intn(t_1), \dots, \intn(t_m))~v_{m+1,1}~\cdots~v_{m+1,j_{m+1}}\quad\cdots\quad v_{m+n,1}~\cdots~v_{m+n,j_{m+n}}} \tpkt
    \end{align*}
    By induction hypothesis on $i = 1,\dots, m$
    we get $\fn(\intn(s_1\sigma), \dots, \intn(s_k\sigma)) \gppv[2\cdot\size{t_i}] \intn(t_i\sigma)$.
    Since $\size{t_i} < \size{t}$, using Lemma~\eref{l:approx}{kmon}
    we have in particular $\fn(\intn(s_1\sigma), \dots, \intn(s_k\sigma)) 
    \gppv[2\cdot\size{t} - 2] \intn(t_i\sigma)$. 
    Using this, $\fn \sp \gn$, and $m < \size{t}$ we conclude 
    \begin{equation}
      \label{eq:egsq:gtv}
     \fn(\intn(s_1\sigma), \dots, \intn(s_k\sigma)) \cppv{ia}[2\cdot\size{t} -1] \gn(\intn(t_1), \dots, \intn(t_m))\tpkt
    \end{equation}
    By induction hypothesis on safe argument positions of $g$ we get 
    $$
    \fn(\intn(s_1\sigma), \dots, \intn(s_k\sigma)) \gppv[2\cdot\size{t_i}] \lst{v_{i,1}~\cdots~v_{i,j_i}} = \ints(t_i\sigma)
    $$
    for all $i =m+1,\dots,m+n$.
    Using a simple inductive argument
    one verifies 
    \begin{equation}
      \label{eq:egsq:safe}
      \fn(\intn(s_1\sigma), \dots, \intn(s_k\sigma)) \gppv[2\cdot\size{t}-1] v_{i,j} \text{ for all $i = m+1,\dots,m+n$ and $j = 1,\dots,j_i$}
    \end{equation}
    from this.
    Let $\intq \in \{\ints, \intn\}$.
    Observe 
    \begin{align*}
      \len(\intq(t\sigma)) 
      & \leqslant 2\cdot\size{t} + \max \{ \norm{s_1\sigma}, \dots, \norm{s_k\sigma}\}
      && \text{by Lemma~\eref{l:int:len}{gsq}} \\
      & \leqslant 2\cdot\size{t} + \width(\fn(\intn(s_1\sigma), \dots, \intn(s_k\sigma))) \tpkt
    \end{align*}
    Using this, Equations~\eqref{eq:egsq:gtv}, Equations~\eqref{eq:egsq:safe}
    and possibly $\fn(\intn(s_1\sigma), \dots, \intn(s_k\sigma)) \cppv{ia}[2\cdot\size{t}] \theconst$
    we have $u \cppv{ialst}[2\cdot\size{t}] \intq(t\sigma)$ as desired.

  \end{description}
  We conclude this auxiliary lemma.
\end{proof}

\begin{lemma}\label{l:embed:root}
  Suppose $s = f(\pseq{s}) \in \BASICS$, $t \in \TERMS$ and $\ofdom{\sigma}{\VS \to \Val}$. 
  Then for predicative interpretation $\intq \in \{\ints, \intn\}$
  we have 
  $$
   s \gpop t \quad \IImp \quad \intq(s\sigma) \gpopv[2\cdot\size{t}] \intq(t\sigma) \tpkt
  $$
\end{lemma}
\begin{proof}
  Let $s$, $t$, $\sigma$ be as given in the lemma.
  We prove the stronger assertions
  \begin{enumerate}
  \item \label{l:er1} $\fn(\intn(s_1\sigma), \dots, \intn(s_k\sigma)) \gpopv[2\cdot\size{t}] \ints(t\sigma)$, 
  \item \label{l:er2} $\fn(\intn(s_1\sigma), \dots, \intn(s_k\sigma)) \gppv[2\cdot\size{t}] \ints(t\sigma)$ if $t \in \termsbelow[\Fun(s)]$, and
  \item \label{l:er3} $\fn(\intn(s_1\sigma), \dots, \intn(s_k\sigma)) \append \NM{s\sigma} \gpopv[2\cdot\size{t}] \intn(t\sigma)$.
  \end{enumerate}
  We continue with the proof of the assertions by induction on $\gpop$.
  \begin{description}[leftmargin=0.3cm]
  \item[{\dcase{$s \cpop{st} t$}}]
    Exactly as in Lemma~\ref{l:embedgsq:root} we conclude $s_i\sigma \esuperterm t\sigma$ and $t\sigma \in \Val$.
    The latter implies $\ints(t\sigma) = \nil$ and thus
    Assertion~\ref{l:er1} and Assertion~\ref{l:er2} are immediate.
    For Assertion~\ref{l:er3}, observe that 
    $\len(\NM{t\sigma}) 
    = \norm{t\sigma} \leqslant \norm{s_i\sigma} 
    \leqslant \width(\fn(\intn(s_1\sigma), \dots, \intn(s_k\sigma)) \append \NM{s\sigma})$
    where the latter inequality is obtained by case analysis on $i$.
    From this and $\fn(\intn(s_1\sigma), \dots, \intn(s_k\sigma)) \cppv{ia}[2\cdot\size{t} - 1] \theconst$ we get
    $\fn(\intn(s_1\sigma), \dots, \intn(s_k\sigma)) \append \NM{s\sigma} \cpopv{ms}[2\cdot\size{t}] \NM{t\sigma} = \intn(t\sigma)$ as desired.

  \item[\dcase{$s \cpop{ia} t$}]
    The assumption gives $t = g(\pseq[m][n]{t})$ where $f \sp g$
    and further $s \gsq t_i$ for all normal argument positions $i = 1,\dots,m$ and
    $s \gpopps t_i$ for all safe argument positions $i = m+1,\dots,m+n$ of $g$. 
    Additionally $t_{i_0} \not\in \termsbelow[\Fun(s)]$ for at most one argument position $i_0$.

    We first verify Assertion~\ref{l:er1} and Assertion~\ref{l:er3} for the non-trivial case $t \not \in \Val$.
    Set $v \defsym \gn(\intn(t_1\sigma), \dots, \intn(t_m\sigma))$ and 
    let $\ints(t_i\sigma) = \lst{v_{i,1}~\cdots~v_{i,j_i}}$ for all safe argument positions 
    $i = m+1,\dots,m+n$, hence
    \begin{align*}
    \ints(t\sigma) =  \lst{\gn(\intn(t_1\sigma), \dots, \intn(t_m\sigma))~v_{m+1,1}~\cdots~v_{m+1,j_{m+1}}\quad\cdots\quad v_{m+n,1}~\cdots~v_{m+n,j_{m+n}}} \tpkt
    \end{align*}
    Applying Lemma~\ref{l:embedgsq:root} on all normal arguments of $t$ we see
    \begin{align}
      \label{e:c3:0}
      \fn(\intn(s_1\sigma), \dots, \intn(s_k\sigma)) \gppv[2\cdot\size{t} - 1] \gn(\intn(t_1\sigma), \dots, \intn(t_m\sigma))
    \end{align}
    from the assumptions $\fn \sp \gn$ and $s \gsq t_i$ for all $i = 1,\dots,m$.
    Since $s \gpop t_{i_0}$ by assumption, induction hypothesis on $i_0$ gives 
    \begin{align*}
      \fn(\intn(s_1\sigma), \dots, \intn(s_k\sigma)) \gpopv[2\cdot \size{t_{i_0}}] \lst{v_{i_0,1}, \dots, v_{i_0,j_{i_0}}} = \ints(t_{i_0}\sigma)\tpkt
    \end{align*}
    Employing $2\cdot\size{t_{i_0}} \leqslant 2\cdot\size{t} - 1$, 
    it is not difficult to check that due to the above inequality we have 
    \begin{align}
      \fn(\intn(s_1\sigma), \dots, \intn(s_k\sigma)) & \gpopv[2\cdot\size{t}-1] {v_{i_0,j_0}} 
      && \text{for some $j_0 \in \{1,\dots,j_{i_0}\}$}\label{e:c3:1}\\
      \fn(\intn(s_1\sigma), \dots, \intn(s_k\sigma)) & \gppv[2\cdot\size{t}-1] {v_{i_0,j}} 
      && \text{for all $j = 1,\dots,j_{i_0}$, $j \not=j_0$.}\label{e:c3:2}
    \end{align}
    Similar induction hypothesis on safe argument positions $i = m+1,\dots,m+n$ ($i \not=i_0$) of $g$,
    where in particular $t_i \in \termsbelow[\Fun(s)]$ by assumption,
    gives 
    \begin{align}
      \fn(\intn(s_1\sigma), \dots, \intn(s_k\sigma)) & \gppv[2\cdot\size{t}-1] {v_{i,j}} 
      && \text{for all $i = m+1,\dots,m+n$, $i \not=i_0$ and $j = 1,\dots,j_{i}$.} \label{e:c3:3}
    \end{align}
    Observe $\len(\ints(t\sigma)) \leqslant \size{t}$ by Lemma~\eref{l:int:len}{S}.
    Summing up, Assertion~\ref{l:er1} follows by $\cpopv{ialst}[2\cdot\size{t}]$
    using Equations \eqref{e:c3:0}, \eqref{e:c3:1}, \eqref{e:c3:2} and \eqref{e:c3:3}.
    Likewise, Assertion~\ref{l:er3} follows using additionally 
    $\fn(\intn(s_1\sigma), \dots, \intn(s_k\sigma)) \cppv{ia}[2\cdot\size{t}-1] \theconst$
    and
    \begin{align*}
      \len(\intn(t\sigma)) 
      & \leqslant 2\cdot\size{t} + \max \{ \norm{s_1\sigma}, \dots, \norm{s_k\sigma}, \norm{s\sigma}\}
      && \text{by Lemma~\eref{l:int:len}{gpop}} \\
      & \leqslant 2\cdot\size{t} + \width(\fn(\intn(s_1\sigma), \dots, \intn(s_k\sigma)) \append \NM{s\sigma}) \tpkt
    \end{align*}

    Finally, for Assertion~\ref{l:er2} we proceed exactly as above, but strengthen the inequality \eqref{e:c3:1} to
    $\fn(\intn(s_1\sigma), \dots, \intn(s_k\sigma)) \gppv[2\cdot\size{t}-1] {v_{i_0,j_0}}$
    which follows as $t_{i_0} \in \termsbelow[\Fun(s)]$ by assumption, and thus
    induction hypothesis can be strengthened to
    $\fn(\intn(s_1\sigma), \dots, \intn(s_k\sigma)) \gppv[2\size{t_{i_0}}] \ints(t_{i_0}\sigma)$.
    This concludes the case $\cpop{ia}$.
    

  \item[\dcase{$s \cpop{ep} t$}]
    Then $t = g(\pseq[m][n]{t})$ where $f \ep g$.
    Further, the assumption gives
    $\mset{\seq[k]{s}} \gpopmul \mset{\seq[m]{t}}$
    and $\mset{\seq[k+l][k+1]{s}} \geqpopmul \mset{\seq[m+n][m+1]{t}}$.
    Hence $t \not \in \termsbelow[\Fun(s)]$ and Property~\ref{l:er2} is vacuously 
    satisfied. 
    We prove Property~\ref{l:er1} and Property~\ref{l:er3}.
    Using that $s_i \in \Val$ for all normal argument positions $i = 1,\dots,m$
    and employing Lemma~\ref{l:gpop:val} we see exactly as in the 
    case $s \cpop{st} t$ above that
    $\mset{\seq[k]{s}} \gpopmul \mset{\seq[m]{t}}$ implies
    $\mset{\intn(s_1\sigma), \dots, \intn(s_k\sigma)} \mextension{\gpopv[2\cdot\size{t}-1]} \mset{\intn(t_1\sigma), \dots, \intn(t_m\sigma)}$.
    Hence 
    \begin{equation}
      \label{eq:root:ep}
      \fn(\intn(s_1\sigma), \dots, \intn(s_k\sigma)) \cpopv{ep}[2\cdot\size{t} -1] \gn(\intn(t_1\sigma), \dots, \intn(t_m\sigma))
    \end{equation}
    follows as by assumption $\fn \ep \gn$ and clearly $m \leqslant \size{t} \leqslant 2\cdot\size{t} - 1$.
    Note that the assumption $\mset{\seq[k+l][k+1]{s}} \geqpopmul \mset{\seq[m+n][m+1]{t}}$ together with
    $s_i \in \Val$ for all $i = k+1,k+l$ gives
    $t_j \in \Val$, and consequently $\ints(t_j\sigma) = \nil$ for all
    $j=k+1,\dots,k+l$.
    Hence 
    $$
    \fn(\intn(s_1\sigma), \dots, \intn(s_k\sigma)) \cpopv{ialst}[2\cdot\size{t}] 
    \lst{\gn(\intn(t_1\sigma), \dots, \intn(t_m\sigma))}
    = \ints(t\sigma) 
    $$
    which concludes Assertion~\ref{l:er1}.
    
    To prove Assertion~\ref{l:er3}
    we additionally verify $\norm{t\sigma} \leqslant \norm{s\sigma}$ 
    by case analysis on $\norm{t\sigma}$.
    Thus $\NM{s\sigma} \geqpopv[2\cdot\size{t}] \NM{t\sigma}$ follows.
    Using this and Equation~\eqref{eq:root:ep} we obtain 
    $$\fn(\intn(s_1\sigma), \dots, \intn(s_k\sigma)) \append \NM{s\sigma} 
    \cpopv{ms}[2\cdot\size{t}]  \gn(\intn(t_1\sigma), \dots, \intn(t_m\sigma)) \append \NM{t\sigma} = \intn(t\sigma)$$
    by Lemma~\eref{l:approx}{kmon} and Lemma~\eref{l:approx}{subseq}.
  \end{description}
  We conclude the lemma.
\end{proof}

\begin{lemma}
  \label{l:embed:ctx}
  Let $\ell \geqslant \max\{\ar(\fn) \mid \fn \in \FSn \}$ and $s,t \in \TERMS$. Then for $\intq \in \{\intn,\ints\}$, 
  $$
  \intq(s) \gpopv[\ell] \intq(t) \quad \IImp \quad \intq(C[s]) \gpopv[\ell] \intq(C[t]) \tpkt
  $$
\end{lemma}
\begin{proof}
  It suffices to consider the inductive step. 
  Consider terms $s = f(s_1,\dots,s_i, \dots, s_{k+l})$ and $t = f(s_1,\dots,t_i, \dots, s_{k+l})$.
  We show that for $\intq \in \{\intn,\ints\}$, under the assumption $\intq(s_i) \gpopv[\ell] \intq(t_i)$ also 
  $\intq(f(s_1,\dots,s_i, \dots, s_{k+l})) \gpopv[\ell] \intq(f(s_1,\dots,t_i, \dots, s_{k+l}))$ holds.
  \begin{description}[leftmargin=0.3cm]
    \item[\dcase{$\intq = \ints$}] 
      Consider the non-trivial case $t \not \in \Val$.
      Without loss of generality, suppose the first $k$ 
      argument positions of $f$ are normal, 
      and the remaining $l$ positions are safe.
      Depending on the position $i$, we distinguish two cases.
      If $i \in \{k+1,\dots,k+l\}$ is safe, then by definition
      \begin{align*}
        \ints(s) & = \lst{\fn(\intn(s_1), \dots, \intn(s_k))} 
                         \append \ints(s_{k+1}) \append \cdots \append \ints(s_i) \append \cdots \append \ints(s_{k+l}) \text{, and} \\
        \ints(t) & = \lst{\fn(\intn(s_1), \dots, \intn(s_k))} 
                         \append \ints(s_{k+1}) \append \cdots \append \ints(t_i) \append \cdots \append \ints(s_{k+l})
      \end{align*}
      If $i$ is a normal argument position, the assumption 
      $\intn(s_i) \gpopv[\ell] \intn(t_i)$ and $\ell \geqslant k$ gives
      \begin{align*}
        \fn(\intn(s_1), \dots, \intn(s_i) ,\dots, \intn(s_k))
        \cpopv{ms}[\ell] \fn(\intn(s_1), \dots, \intn(t_i) ,\dots, \intn(s_k))
      \end{align*}
      and the lemma follows again using Lemma~\eref{l:approx}{subseq}.

    \item[\dcase{$\intq = \intn$}] 
      Recall that $\intn(s) = \ints(s) \append \NM{s}$ and $\intn(t) = \ints(t) \append \NM{t}$.
      If $\norm{s} \geqslant \norm{t}$ then $\intn(s) \gpopv[l] \intn(t)$ follows 
      from $\ints(s) \gpopv[l] \ints(t)$ and Lemma~\eref{l:approx}{subseq}.
      Hence suppose $\norm{s} < \norm{t}$.
      %
      First, consider the more involved case $t \not \in \Val$.
      As $\norm{s} < \norm{t}$ implies that $i$ is a safe argument position of $f$, we thus have
      \begin{align*}
        \intn(s) & = \lst{\fn(\intn(s_1), \dots, \intn(s_k))} 
                         \append \ints(s_{k+1}) \append \cdots \append \ints(s_i) \append \cdots \append \ints(s_{k+l}) \append \NM{s} \text{, and} \\
        \intn(t) & = \lst{\fn(\intn(s_1), \dots, \intn(s_k))} 
                         \append \ints(s_{k+1}) \append \cdots \append \ints(t_i) \append \cdots \append \ints(s_{k+l})  \append \NM{t} 
      \end{align*}
      Using Lemma~\eref{l:approx}{modeqi}
      and Lemma~\eref{l:approx}{subseq}
      we see that $\intn(s) \gpopv[\ell] \intn(t)$ follows from 
      $\ints(s_i) \append \NM{s} \gpopv[\ell] \ints(t_i) \append \NM{t}$.
      Note $\norm{s_i} < \norm{s}$
      and observe that 
      the assumption $\norm{s} < \norm{t}$ gives $\norm{t} = \norm{t_i} + 1$ 
      by the shape of $s$ and $t$.
      Thus using Lemma~\eref{l:approx}{subseq} and the assumption $\intn(s_i) \gpopv[\ell] \intn(t_i)$
      we can even prove the stronger property 
      $\ints(s_i) \append \NM{s_i} \append \theconst \gpopv[\ell] \ints(t_i) \append \NM{t_i} \append \theconst = \intn(t)$.
      
      By similar reasoning we can also prove $t \in \Val$ where $\intn(t) = \norm{t}$.
      This concludes the case analysis.
    \end{description}
\end{proof}

We have established the embedding for completely defined TRSs.\@ 
Putting things together we obtain:
This allows us to estimate the derivation height in terms of $\Slow[\ell]$.
\begin{lemma}\label{l:embed:bound}
  Let $\RS$ be a completely defined TRS compatible with $\gpop$. 
  Define $\ell \defsym \max\{\ar(\fn) \mid \fn \in \FSn \} \cup \{2\cdot\size{r} \mid {l \to r} \in \RS \}$
  and $\intq \in \{\intn,\ints\}$. Then $\dheight(s, \irew) \leqslant \Slow[\ell](\intq(s))$.
\end{lemma}
\begin{proof}
  Suppose $\RS$ is completely defined TRS compatible with $\gpop$, and
  let $\ell$ be given by the Lemma.
  Consider a maximal $\RS$-derivation
  $$
  \mparbox[c]{1cm}{s} 
  \mparbox[c]{1cm}{\irew[\RS]} \mparbox[c]{1cm}{s_1} 
  \mparbox[c]{1cm}{\irew[\RS]} \mparbox[c]{1cm}{s_2} 
  \mparbox[c]{1cm}{\irew[\RS]} \mparbox[c]{1cm}{\cdots}
  \mparbox[c]{1cm}{\irew[\RS]} \mparbox[c]{1cm}{s_m\tkom} 
  $$
  starting from an arbitrary term $s$, i.e., $m = \dheight(s, \irew[\RS])$.
  Using Lemma~\ref{l:embed:root} together with Lemma~\ref{l:embed:ctx} $m$-times 
  we get 
  $$
   \mparbox[c]{1cm}{\intq(s)} 
   \mparbox[c]{1cm}{\gpopv[\ell]} \mparbox[c]{1cm}{\intq(s_1)} 
   \mparbox[c]{1cm}{\gpopv[\ell]} \mparbox[c]{1cm}{\intq(s_2)}
   \mparbox[c]{1cm}{\gpopv[\ell]} \mparbox[c]{1cm}{\cdots}
   \mparbox[c]{1cm}{\gpopv[\ell]} \mparbox[c]{1cm}{\intq(s_m)}
  $$
  and consequently $m \leqslant \Slow[\ell](\intq(s))$ by definition.
\end{proof}

The final proof step is to lift the requirement that $\RS$ is completely defined.
Call a normal form $s$ \emph{garbage} if its root symbol is defined.
Let $\bot \not \in \FS$ be a fresh constructor symbol.
For each garbage term $s$ we extend $\RS$ by a rule 
that replaces $s$ with $\bot$. Although infinite, the resulting 
system is completely defined.

\begin{definition}\label{d:rss}
  Let $\bot$ be a fresh constructor symbol $\bot \not \in \FS$ and $\RS$ a TRS over $\FS$.
  We define $\RSS$ over the signature $\FS \cup \{\bot\}$ by
  $$
  \RSS \defsym \{t \to \bot \mid \text{$t \in \TA(\FS \cup \{\bot\},\VS)$ is a normal form of $\RS$ with defined root symbol} \}\tpkt
  $$
  We set $\RSbot \defsym \RS \cup \RSS$.
\end{definition}
We extend the precedence $\qp$ on $\FS$ to $\FS \cup \{\bot\}$ so that $\bot$ is minimal.
As clearly $s \cpop{ia} \bot$ for each garbage term $s$, 
for predicative TRS $\RS$ the TRS $\RSS$ is compatible with $\gpop$.
Note that $\RSS$ is confluent and terminating, in particular every term
$s$ has a unique normal form with respect to $\RSS$, in notation $\normalise{s}$.
Clearly $\normalise{f(\seq{s})} = \normalise{f(\normalise{s_1},\dots, \normalise{s_n})}$.
Exploiting that the additional rules do not interfere with 
pattern matching of $\RS$, the TRS $\RSbot$ is able to simulate $\RS$ in the following sense. 

\begin{lemma}\label{l:rss:step}
  Suppose $\RS$ is a constructor TRS. Then
  $$
  s \irew t \quad\IImp\quad \normalise{s} \irst[\RSbot] \normalise{t}
  $$
\end{lemma}
\begin{proof}
  Suppose $s \irew t$, i.e., $s = C[f(l_1\sigma, \dots, l_n\sigma)]$ and 
  $t = C[r\sigma]$ for some context $C$, rule ${f(\seq{l}) \to r} \in \RS$
  and substitution $\sigma$ where $l_i \sigma \in \NF(\RS)$ for all $i = 1,\dots,n$.
  We continue by induction on $C$.
  Let $\sigma_{\normalise{}}(x) \defsym \normalise{x\sigma}$ for all $x \in \dom(\sigma)$.

  Consider the base case $C = \hole$. Since $\RS$ is by assumption a constructor TRS, 
  the direct arguments of the left-hand sides of $\RS$ do not contain defined symbols, 
  consequently $\normalise{l_i}\sigma = \normalise{l_i\sigma_{\normalise{}}} = l_i\sigma_{\normalise{}}$
  is a constructor term
  for all $i = 1,\dots,n$.
  We conclude the inductive step
  $$
  \normalise{f(l_1\sigma, \dots, l_n\sigma)} = f(l_1\sigma_{\normalise{}}, \dots, l_n\sigma_{\normalise{}})
  \irew[\RSbot] r{\sigma_{\normalise{}}} 
  \irss[\RSbot] \normalise{(r\sigma)} \tpkt
  $$
  Here in the first equality we employ that  $f(l_1\sigma_{\normalise{}}, \dots, l_n\sigma_{\normalise{}})$ is not a normal form of $\RS$.

  For the inductive step, let $s = f(s_1, \dots, s_i, \dots, s_n)$ and 
  $t = f(s_1, \dots, t_i, \dots, s_n)$ where $s_i \irew t_i$. Induction hypothesis 
  gives $\normalise{s_i} \irew[\RSbot] \normalise{t_i}$. 
  Then 
  $$
  \normalise{s} = f(\normalise{s_1}, \dots, \normalise{s_i}, \dots, \normalise{s_n}) 
  \irew[\RSbot] f(\normalise{s_1}, \dots, \normalise{t_i}, \dots, \normalise{s_n}) 
  \irss[\RSbot] \normalise{t} \tpkt
  $$
  For the first equality we employ that 
  $\normalise{s_i} \not \in \NF(\RS)$. This concludes the proof.
\end{proof}

An immediate consequence is the following.
\begin{lemma}\label{l:rss:simul}
  Let $\RS$ be a predicative recursive TRS. Then
  $\RSbot$ is a completely defined TRS compatible with $\gpop$. 
  Further $\dheight(s, \irew) \leqslant \dheight(s, \irew[\RSbot])$ 
  for all basic terms $s$.
\end{lemma}
\begin{proof}
  We have already observed that $\RSbot$ is compatible with $\gpop$. 
  Moreover it is completely defined by definition.
  As $\RS$ is predicative recursive, it is a constructor TRS.
  To prove the second halve of the assertion, consider a maximal derivation
  $$
  \mparbox[c]{1cm}{s} \mparbox[c]{1cm}{\irew[\RS]}
  \mparbox[c]{1cm}{s_1}  \mparbox[c]{1cm}{\irew[\RS]}
  \mparbox[c]{1cm}{s_2}  \mparbox[c]{1cm}{\irew[\RS]}
  \mparbox[c]{1cm}{\cdots}  \mparbox[c]{1cm}{\irew[\RS]}
  \mparbox[c]{1cm}{s_m}
  $$
  starting from a basic term $s$, i.e., $m = \dheight(s, \irew)$.
  If $m = 0$ the lemma is immediate.
  For the case $m > 0$, 
  $m$-times application of Lemma~\ref{l:rss:step} gives
  $$
  \mparbox[c]{1cm}{\normalise{s}}  \mparbox[c]{1cm}{\irew[\RS]}
  \mparbox[c]{1cm}{\normalise{s_1}}  \mparbox[c]{1cm}{\irew[\RS]}
  \mparbox[c]{1cm}{\normalise{s_2}}  \mparbox[c]{1cm}{\irew[\RS]}
  \mparbox[c]{1cm}{\cdots}  \mparbox[c]{1cm}{\irew[\RS]}
  \mparbox[c]{1cm}{\normalise{s_m}} \tpkt
  $$
  
  Hence overall, $\dheight(s, \irew) \leqslant \dheight(\normalise{s}, \irew[\RSbot])$.
  Since by assumption $s$ is a basic term not in normal form, 
  we have $\normalise{s} = s$ and the lemma follows.
\end{proof}
We arrive at the proof of the main theorem:

\begin{proof}[Proof of Theorem~\ref{t:popstar}]
  Let $\RS$ be a predicative recursive TRS and fix an 
  arbitrary basic term $s = f(\pseq[m][n]{u})$. 
  Set $\ell \defsym \max\{\ar(\fn) \mid \fn \in \FSn \} \cup \{2\cdot\size{r} \mid {l \to r} \in \RSbot\}$
  and note that $\ell$ is well defined since $\FSn$ and $\RS$ are finite. 
  Putting things together we see
  \begin{align*}
    \dheight(s, \irew) & \leqslant \dheight(s, \irew[\RSbot])
    && \text{using Lemma~\ref{l:rss:simul}} \\
    & \leqslant \Slow[\ell](\ints(s)) 
    && \text{using Lemma~\ref{l:embed:bound}}\\
    & \in \bigO\bigl((\max_{i=1}^m \depth(u_i))^{d}\bigr) 
    && \text{using Corollary~\ref{c:pop}}
  \end{align*}
  where $d$ depends only on $\ell$.
\end{proof}


\section{An Order-Theoretic Characterisation of the Polytime Functions}\label{s:icc}
We now present the application of polynomial path orders
in the context of \emph{implicit computational complexity (ICC)}.
As by-product of Proposition~\ref{p:invariance} and Theorem~\ref{t:popstar} we immediately obtain
that $\POPSTAR$ is \emph{sound} for $\FNP$ respectively $\FP$.
\begin{theorem}\label{t:icc:soundness}
  Let $\RS$ be a predicative recursive TRS.\@ 
  For every relation $\sem{f}$ defined by $\RS$, 
  the functional problem $F_f$ associated with $\sem{f}$ is in $\FNP$.
  Moreover, if $\RS$ is confluent than $\sem{\mf} \in \FP$.
\end{theorem}
Although it is decidable whether a TRS $\RS$ is predicative recursive (we 
present a sound and complete automation in Section~\ref{s:exps}), 
confluence is undecidable in general. To get a decidable result for $\FP$, 
one can replace by an decidable criteria, for instance orthogonality. 

We will now also establish that \POPSTAR\ is \emph{complete} for $\FP$, that is, 
every function $f \in \FP$ is computed by some confluent (even orthogonal) predicative recursive TRS.\@
For this we use the \emph{term rewriting formulation} of the 
predicative recursive functions from~\cite{BW96}. 

\begin{definition}\label{d:Rb}
For each $k,l \in \N$ the set of function symbols $\Fb^{k,l}$ 
with $k$ normal and $l$ safe argument positions is the 
least set of function symbols such that
\begin{enumerate}
\item $\epsilon \in \Fb^{0,0}$, $\mS_1,\mS_2 \in \Fb^{0,1}$, $\m{P} \in \Fb^{0,1}$, $\m{C} \in \Fb^{0,4}$
  and $\m{I}^{k,l}_j, \m{O}^{k,l} \in \Fb^{k,l}$, where $j = 1, \dots, k+l$; 
\item if $\vec{r} = \seq[m]{r} \in \Fb^{k,0}$, $\vec{s} = \seq[n]{s} \in \Fb^{k,l}$
  and $h \in \Fb^{m,n}$ then $\m{SC}[h,\vec{r}, \vec{s}] \in \Fb^{k,l}$; 
\item if $g \in \Fb^{k,l}$ and $h_1,h_2 \in \Fb^{k+1,l+1}$ then $\m{SRN}[g,h_1,h_2] \in \Fb^{k+1,l}$; 
\end{enumerate}
The \emph{predicative signature} is given by $\Fb \defsym \bigcup_{k,l \in \N} \Fb^{k,l}$.
Only the constant $\epsilon$ and \emph{dyadic successors} $\mS_1,\mS_2$, 
which serve the purpose of encoding natural numbers in binary, 
are constructors.
The remaining symbols from $\Fb$ are defined 
by the following (infinite) schema of rewrite rules $R_\B$.
Here we $k,l$ range over $\N$ and we abbreviate 
$\vec{x} = \seq[k]{x}$ and $\vec{y} = \seq[l]{y}$ for $k$ respectively $l$ distinct variables. 
\medskip

\begin{tabular}{@{\quad}r@{~}c@{~}l@{\hspace{-10mm}}r}
  \multicolumn{4}{@{}l}{\textbf{Initial Functions}}\\[2mm]
  $\m{P}(\sn{}{\epsilon})$ & $\to$ & $\epsilon$ \\[1mm]
  $\m{P}(\sn{}{\mS_i(\sn{}{x})})$ & $\to$ & $x$ &  for $i=1,2$ \\[1mm]
  $\m{I}^{k,l}_j(\svec{x}{y})$ & $\to$ & $ x_j$ & for all $j = 1,\dots,k$ \\[1mm]
  $\m{I}^{k,l}_j(\svec{x}{y})$ & $\to$ & $y_{j-k}$ & for all $j =k+1, \dots, l+k$ \\[1mm]
  $\m{C}(\sn{}{\epsilon, y, z_1, z_2})$ & $\to$ & $y$ \\[1mm]
  $\m{C}(\sn{}{\mS_i(\sn{}{x}), y, z_1, z_2})$ & $\to$ & $z_i$ & for $i = 1, 2$ \\[1mm]
  $\m{O}(\svec{x}{y})$ & $\to$ & $\epsilon$ 
  \\[3mm]
  \multicolumn{4}{@{}l}{\textbf{Safe Composition} ($\m{SC}$)} \\[2mm]
    $\m{SC}[h,\vec{r}, \vec{s}](\svec{x}{y})$ & $\to$ & $h(\sn{\vec{r}(\sn{\vec{x}}{})}{\vec{s}(\svec{x}{y})})$
  \\[3mm]
  \multicolumn{4}{@{}l}{\textbf{Safe Recursion on Notation} ($\m{SRN}$)} \\[2mm]
    $\m{SRN}[g,h_1,h_2](\sn{\epsilon, \vec x}{\vec y})$ & $\to$ & $ g(\svec{x}{y})$ \\[1mm]
    $\m{SRN}[g,h_1,h_2](\sn{\mS_i (\sn{}{z}), \vec x}{\vec y}) 
    $ & $\to$ & $h_i(\sn{z, \vec x}{\vec y, \m{SRN}[g,h_1,h_2](\sn{z, \vec x}{\vec y})})$ & for $i = 1, 2$
\end{tabular}
\end{definition}

We emphasise that the above rules are all orthogonal.
Also, we stress that the system $R_\B$ is dupped \emph{infeasible} in~\cite{BW96}.
Indeed $R_\B$ admits an exponential lower bound on the derivation height
which has to do with effects caused by duplicating redexes as
explained already in Example~\ref{ex:dup} on page~\pageref{ex:dup}.
Therefore $R_\B$ is not (directly) suitable as a term-rewriting characterisation
of the predicative recursive functions. 
However this exponential lower-bound is only correct if we consider unrestricted rewriting.
The following proposition verifies that $R_\B$ generates only polytime computable functions.
\begin{proposition}\label{prop:Rf}{\cite[Lemma~5.2]{BW96}}
  Let $f \in \FP$. There exists a finite restriction $\RS_f \subsetneq R_\B$
  such that $\RS_f$ computes $f$.
\end{proposition}
We arrive at our completeness result.
\begin{theorem}\label{t:icc:completeness}
  For every $f \in \FP$ there exists an orthogonal predicative recursive
  TRS $\RS_f$ that computes $f$.
\end{theorem}
\begin{proof}
  Take the TRS $\RS_f \subsetneq R_\B$ from Proposition~\ref{prop:Rf} that computes $f$.
  Obviously $\RS_f$ is orthogonal hence confluent, 
  it remains to verify that $\RS_f$ is compatible with some instance $\gpop$.
  To define $\gpop$ we use the separation of normal from safe argument positions
  as indicated in the rules.
  To define the precedence underlying $\gpop$, we 
  first define a mapping $\lh$ from the signature of $\Fb$ 
  into the natural numbers as follows:
  \begin{itemize}
  \item $\lh(f) \defsym 0$ if $f$ is one of $\epsilon$, $\mS_0$, $\mS_1$, $\m{C}$, $\m{P}$, $\m{I}^{k,l}_j$ or $\m{O}^{k,l}$;
  \item $\lh(\m{SC}[h,\vec{r}, \vec{s}]) \defsym 1 + \lh(h) + \sum_{r \in \vec{r}} \lh(r) + \sum_{s \in \vec{s}}  \lh(s)$;
  \item $\lh(\m{SRN}[g,h_1,h_2]) \defsym 1 + \lh(g) + \lh(h_1) + \lh(h_2)$. 
  \end{itemize}
  Finally for each pair of function symbol $f$ and $g$ occurring in $\RS_f$ set
  $f \sp g$ if and only if $\lh(f) > \lh(g)$.
  Then $\sp$ defines an admissible precedence. 
  It is straight forward to verify that $\RS_f \subseteq {\gpop}$ where
  $\gpop$ is based on the precedence $\sp$ and the safe mapping as indicated in
  Definition~\ref{d:Rb}.
\end{proof}

By Theorem~\ref{t:icc:soundness} and Theorem~\ref{t:icc:completeness} we thus obtain a
precise characterisation of the class polytime computable functions.

\begin{corollary}\label{c:FP}
  The class of confluent (or orthogonal) predicative recursive TRSs define exactly $\FP$.
\end{corollary}


\section{A Non-Trivial Closure Property of the Polytime Computable Functions}\label{s:popstarps}

Bellantoni already observed that the class $\B$ is closed under
\emph{predicative recursion on notation with parameter substitution} (scheme \eqref{scheme:srnps}). 
Essentially this recursion scheme allows substitution on \emph{safe} argument positions. More precise, 
a new function $f$ is defined by the equations
\begin{equation}\label{scheme:srnps} \tag{\ensuremath{\mathsf{SRN_{PS}}}}
  \begin{array}{r@{\;}c@{\;}l}
   f(\sn{0,\vec{x}}{\vec{y}}) & = & g(\sn{\vec{x}}{\vec{y}}) \\
   \qquad\quad f(\sn{2z + i,\vec{x}}{\vec{y}}) & = &
   h_i(\sn{z,\vec{x}}{\vec{y},f(\sn{z,\vec{x}}{\vec{p}(\svec{x}{y})})}),~i \in \set{1,2} \tpkt
  \end{array}
\end{equation}
Notably closure of $\B$ under parameter substitution has been proven also been Beckmann and Weiermann~\cite{BW96}
based on rewriting techniques.
In the following we introduce a polynomial path order beyond MPO, 
the \emph{polynomial path order with parameter substitution} (\emph{\POPSTARP} for short).
The next definition introduces $\POPSTARP$.
It is a variant of $\POPSTAR$ where 
clause $\cpop{ep}$ has been modified and allows computation at safe argument positions.
\begin{definition}\label{d:gpopps}
  Let $s, t \in \TERMS$ such that $s = f(\pseq[k][l]{s})$.
  Then $s \gpopps t$ with respect to the precedence $\qp$ and safe mapping $\safe$ if either
  \begin{enumerate}
  \item\label{d:gpopps:st} $s_i \geqpopps t$ for some $i \in \{1,\dots,k+l\}$, or
  \item\label{d:gpopps:ia} $f \in \DS$, $t = g(\pseq[m][n]{t})$ where $f \sp g$ 
    and the following conditions hold:
    \begin{itemize}
    \item $s \gsq t_j$ for all normal argument positions $j = 1,\dots,m$;
    \item $s \gpopps t_j$ for all safe argument positions $j = m+1,\dots,m+n$;
    \item $t_j \not\in \TA(\sigbelow{\Fun(s)}{\FS},\VS)$ for at most one safe argument position $j \in \{m+1,\dots,m+n\}$;
    \end{itemize}
  \item\label{d:gpopps:ep} $f \in \DS$, $t = g(\pseq[m][n]{t})$ where $f \ep g$
    and the following conditions hold:
    \begin{itemize}
    \item $\mset{\seq[k]{s}} \gpopmulps \mset{\seq[m]{t}}$;
    \item $s \gpopps t_j$ and $t_j \in \TA(\sigbelow{\Fun(s)}{\FS},\VS)$ for all safe argument positions $j = m+1, \dots, m+n$.
    \end{itemize}
  \end{enumerate}
  Here ${\geqpopps} \defsym {\gpopps \cup \eqis}$.
\end{definition}

We adapt the notion of predicative recursive TRS to \POPSTARP\ in the obvious way.
It is not difficult to see that $\POPSTARP$ extends the analytic power of $\POPSTAR$.
\begin{lemma}\label{l:psextends}
  For any underlying admissible precedence $\qp$, ${\gpop} \subseteq {\gpopps}$.
\end{lemma}
Note that \POPSTARP~is strictly more powerful than \POPSTAR, as witnessed by 
the following example.

\begin{example}\label{ex:rsrev}
Consider the TRS $\RSrev$ defining the reversal of
lists in a tail recursive fashion:
\begin{align*}
  \mrev(\sn{xs}{}) & \to \mrevt(\sn{xs}{\mnil}) &
  \mrevt(\sn{\nil}{ys}) & \to ys &
  \mrevt(\sn{\mcs(x, xs)}{ys}) & \to \mrevt(\sn{xs}{\mcs(x,ys)}) \tpkt
\end{align*}
Then $\RSrev \subseteq {\gpopps}$ with precedence 
$\mrev \sp \mrevt \sp \mnil \ep \mcs$.
Note that orientation of the last rule with $\gpopps$ 
breaks down to $\mcs(x, xs) \gpopps xs$ and  $\mrevt(\sn{\mcs(x, xs)}{ys}) \gpopps \mcs(x,ys)$.
On the other hand, $\gpop$ fails as the corresponding clause $\cpop{ep}$ requires $ys \geqpop \mcs(x,ys)$.
\end{example}

The order \POPSTARP\ is complete for the class of
polytime computable functions.
To show that it is sound, we prove that 
\POPSTARP~induces polynomially bounded runtime complexity in the sense of Theorem~\ref{t:popstar}.
The crucial observation is that the embedding of $\irew$ into $\gpopv$ does not break
if we relax compatibility constraints to $\RS \subseteq {\gpopps}$.

\begin{lemma}\label{l:embed:root:ps}
  Suppose $s = f(\pseq{s}) \in \Tb$, $t \in \TERMS$ and $\ofdom{\sigma}{\VS \to \Val}$. 
  Then for predicative interpretation $\intq \in \{\ints, \intn\}$
  we have 
  $$
   s \gpopps t \quad \IImp \quad \intq(s\sigma) \gpopv[2\cdot\size{t}] \intq(t\sigma) \tpkt
  $$
\end{lemma}
\begin{proof}
  First one verifies that 
  Lemma~\ref{l:int:len} holds even if we replace 
  $\gpop$ by $\gpopps$. 
  In particular, the assumptions give 
  \begin{align}
    \label{l:embed:root:ps:len}
    \len(\intn(t\sigma)) \leqslant 2\cdot\size{t} + \max \{ \norm{s_1\sigma}, \dots, \norm{s_k\sigma}, \norm{s\sigma}\}
  \end{align}
  The proof proceeds then in correspondence to 
  Lemma~\ref{l:embed:root} by induction on
  $\gpopps$. We cover only the new case.
  Let $s$, $t$, $\sigma$ be as given in the lemma.
  \begin{description}[leftmargin=0.3cm]
  \item[{\dcase{$s \cpopps{ep} t$}}]
    Then $t = g(\pseq[m][n]{t})$ where $f \ep g$.
    Further, the assumption gives
    $\mset{\seq[k]{s}} \gpopmul \mset{\seq[m]{t}}$.
    As $t \not \in \termsbelow[\Fun(s)]$ 
    it suffices to verify Property~\ref{l:er1} and Property~\ref{l:er3}
    from Lemma~\ref{l:embed:root}.
    Exactly as in the corresponding case of Lemma~\ref{l:embed:root} we see
    \begin{equation}
      \label{eq:root:ep:ps}
      \fn(\intn(s_1\sigma), \dots, \intn(s_k\sigma)) \cpopv{ep}[2\cdot\size{t} -1] \gn(\intn(t_1\sigma), \dots, \intn(t_m\sigma)) \tpkt
    \end{equation}
    As by assumption $s \gpopps t_j$ and $t_j \in \termsbelow[\Fun(s)]$, induction hypothesis gives 
    \begin{equation}
      \label{eq:ti:ep:ps}
      \fn(\intn(s_1\sigma), \dots, \intn(s_k\sigma)) \gppv[2\cdot\size{t} -1] \ints(t_j\sigma) \tpkt
    \end{equation}
    As $\len(\ints(t\sigma)) \leqslant \size{t}$ by Lemma~\eref{l:int:len}{S}, 
    we obtain $\fn(\intn(s_1\sigma), \dots, \intn(s_k\sigma)) \cpopv{ialst}[2\cdot\size{t}] \ints(t\sigma)$
    from Equation~\eqref{eq:root:ep:ps} and Equation~\eqref{eq:ti:ep:ps}.
    Likewise, from this Assertion~\ref{l:er3} follows by $\cpopv{ms}$ using additionally 
    $\fn(\intn(s_1\sigma), \dots, \intn(s_k\sigma)) \cppv{ia}[2\cdot\size{t}-1] \theconst$
    and
    \begin{align*}
      \len(\intn(t\sigma)) 
      & \leqslant 2\cdot\size{t} + \max \{ \norm{s_1\sigma}, \dots, \norm{s_k\sigma}, \norm{s\sigma}\}
      && \text{by Equation~\eqref{l:embed:root:ps:len}} \\
      & \leqslant 2\cdot\size{t} + \width(\fn(\intn(s_1\sigma), \dots, \intn(s_k\sigma)) \append \NM{s\sigma}) \tpkt
    \end{align*}
  \end{description}  
\end{proof}

Following the Proof of Theorem~\ref{t:popstar}, but replacing 
Lemma~\ref{l:embed:root} by Lemma~\ref{l:embed:root:ps} we obtain:
\begin{theorem}
  Let $\RS$ be predicative recursive TRS (in the sense of Definition~\ref{d:gpopps}).
  Then the innermost derivation height of any basic term 
  $f(\svec{u}{v})$ is bounded by a polynomial in the 
  maximal depth of normal arguments $\vec{u}$.
  The polynomial depends only on $\RS$ and the signature $\FS$.
\end{theorem}

Using this theorem, Proposition~\ref{p:invariance} states
that \POPSTARP\ is sound for the polytime computable functions.
Lemma~\ref{l:psextends} together with Theorem~\ref{t:icc:completeness}
shows completeness of \POPSTARP~for the polytime computable functions.
\begin{corollary}\label{c:fptime:ps}
  The class of confluent (or orthogonal) predicative recursive TRSs (in the sense of Definition~\ref{d:gpopps}) define exactly $\FP$.
\end{corollary}


\section{Automation of Polynomial Path Orders}\label{s:impl}

In this section we present an automation of polynomial path orders,  
for brevity we restrict our efforts to the order $\gpop$.
Consider a constructor TRS $\RS$. Checking whether $\RS$ is predicative 
recursive is equivalent to guessing a precedence $\qp$ and partitioning 
of argument positions so that $\RS \subseteq {\gpop}$ holds for the induces
order $\gpop$. 
As standard for recursive path orders~\cite{ZM07,SFTGACMZ07},
this search can be automated by encode the constraints imposed by Definition~\ref{d:gpop}
into \emph{propositional logic}. 
To simplify the presentation, we extend language of propositional 
logic with truth-constants $\top$ and $\bot$ in the obvious way.
In the constraint presented below we employ the following atoms. 

\paragraph*{\textbf{Propositional Atoms}}
To encode the separation of normal from safe arguments, we
introduce $f \in \DS$ and $i = 1, \dots, \ar(f)$ the atoms $\esafe{f}{i}$
so that $\esafe{f}{i}$ represents the assertion that the \nth{$i$} argument position of $f$ is safe.
Further we set $\esafe{f}{i} \defsym \top$ for $n$-ary $f \in \CS$ and $i = 1,\dots,n$ which
reflects that argument positions of constructors are always safe. 

One verifies that predicative recursive
TRSs are even compatible with $\gpop$ as induced
by an \emph{admissible} precedence where constructors are equivalent, 
that is, polynomial path orders are \emph{blind} on constructors.
This is exploited in the propositional encoding of precedences, where
we encode a precedence $\qp$ on the set of defined symbols $\DS$ only:
For each pair of symbols $f,g \in \DS$, we introduce
propositional atoms $\esp{f}{g}$ and $\eep{f}{g}$ so that $\esp{f}{g}$
represents the assertion $f \sp g$, and likewise 
$\eep{f}{g}$ represents the assertion $f \ep g$.
Overall we define for function symbols $f$ and $g$ the propositional formulas
\begin{equation*}
  \enc{f \sp g} \defsym 
  \begin{cases}
    \top & \text{if $f\in\DS$ and $g\in \CS$,} \\
    \bot & \text{if $f\in\CS$ and $g\in \CS$,} \\
    \esp{f}{g} & \text{otherwise.}
  \end{cases}
  \quad
  \enc{f \ep g} \defsym
  \begin{cases}
    \top & \text{if $f\in\CS$ and $g\in \CS$, } \\
    \bot & \text{if $f\in\DS$ and $g \in \DS$,} \\
    \eep{f}{g} & \text{otherwise.}
  \end{cases}
\end{equation*}

To ensure that the variables $\esp{f}{g}$ and respectively $\eep{f}{g}$ encode a preorder on $\DS$ 
we encode an order preserving homomorphism into the natural order $>$. 
To this extend, to each $f \in \DS$ we associate a natural number $\rk_f$ encoded as binary string
with $\lceil \log_2(\size{\DS}) \rceil$ bits. 
It is straight forward to define Boolean formulas $\enc{\rk_f > \rk_g}$ (respectively $\enc{\rk_f = \rk_g}$) that are satisfiable iff 
the binary numbers $\rk_f$ and $\rk_g$ are decreasing (respectively equal) in the natural order. Using these we set
\begin{align*}
  \vprec(\DS)
  \defsym \bigwedge_{f,g \in \DS} (\enc{f \sp g} \imp \enc{\rk_f > \rk_g})
  \wedge \bigwedge_{f,g \in \DS} (\enc{f \ep g} \imp \enc{\rk_f = \rk_g})
\end{align*}

We say that a propositional assignment $\mu$ \emph{induces} the precedence
$\qp$ if $\mu$ satisfies $\enc{f \sp g}$ when $f \sp g$ and $\enc{f \ep g}$ when $f \ep g$. 
The next lemma verifies that $\vprec$ serves our needs.
\begin{lemma}
  For any valuation $\mu$ that satisfies $\vprec(\DS)$, 
  $\mu$ induces an admissible precedence on $\FS$.
  Vice versa, for any admissible precedence $\qp$ on $\FS$, 
  any valuation $\mu$, satisfying $\mu(\enc{f \sp g})$ iff $f \sp g$ and 
  $\mu(\enc{f \sp g})$ iff $f \ep g$, also satisfies
  the formula $\vprec(\DS)$.
\end{lemma}

\paragraph{\textbf{Order Constraints}}
For concrete pairs of terms $s = f(\seq{s})$ and $t$, we define the order constraints
$$
\enc{s \gpop t} \defsym \enc{s \cpop{st} t} \vee \enc{s \cpop{ia} t} \vee \enc{s \cpop{ep} t}
$$
which enforces the orientation $f(\seq{s}) \gpop t$ using propositional formulations 
of the three clauses in Definition~\ref{d:gpop}. 
To complete the definition for arbitrary left-hand sides, we set $\enc{x \gpop t} \defsym \bot$ for all $x \in \VS$.
Further weak orientation is given by
$$
\enc{s \geqpop t} \defsym \enc{s \gpop t} \vee \enc{s \eqis t} \tkom
$$
where the constraint $\enc{s \eqis t}$ refers to a formulation of 
Definition~\ref{d:eqis} in propositional logic, defined as follows.
For $s = t$ we simply set $\enc{s \eqis t} \defsym \top$. 
Consider the case $s = f(\seq{s})$ and $t = g(\seq{t})$. 
Then $s \eqis t$ if $f \ep g$ 
and moreover $s_i \eqis t_{\pi(i)}$ for all $i = 1,\dots,n$ and 
some permutation $\pi$ on argument positions that takes the separation of normal and safe positions into account.
To encode $\pi(i) = j$, we use fresh atoms $\pi_{i,j}$ for $i,j=1,\dots,n$. 
The propositional formula $\vperm(\pi,n) \defsym \bigwedge_{i=1}^{n} \eone(\pi_{i,1}, \dots, \pi_{i,n})$
is used to assert that the atoms $\pi_{i,j}$ reflect a permutation on $\{1,\dots,n\}$. 
Here $\eone(\pi_{i,1}, \dots, \pi_{i,n})$ expresses that exactly one of its arguments evaluates to $\top$.
We set
\begin{align*}
  \enc{s \eqis t} \defsym 
    \enc{f \ep g} \wedge \vperm(\pi,n)
      \wedge~\bigl({\bigwedge_{j=1}^{n} \pi_{i,j} \imp \enc{s_i \eqis t_j} \wedge (\esafe{f}{i} \iff \esafe{g}{j})}\bigr) \tpkt
\end{align*}
To complete the definition, we set $\enc{s \eqis t} = \bot$ for 
the remaining cases.
\begin{lemma}
  Suppose $\mu$ induces an admissible precedence $\qp$ and satisfies $\enc{s \eqis t}$. 
  Then $s \eqis t$ with respect to the precedence $\qp$.
  Vice versa, if $s \eqis t$ then $\enc{s \eqis t}$ is satisfiable by assignments $\mu$ 
  that induce the precedence underlying $\eqis$.
\end{lemma}

We now define the encoding for the different cases underlying the definition of $\gpop$.
Assuming that $\enc{s_i \geqpop t}$ enforces $s_i \gpop t$ clause $\cpop{st}$ is expressible
as 
\begin{align*}
  \enc{f(\seq{s}) \cpop{st} t} \defsym \bigor_{i=1}^n \enc{s_i \geqpop t}
\end{align*}
in propositional logic. 
For clause $\cpop{ia}$ we use propositional atoms $\alpha_i$ ($i = 1,\dots,m$)
to mark the unique argument position of $t = g(\seq[m]{t})$ that allows
$t_i \not\in \termsbelow[\Fun(s)]$. 
The propositional formula $\ezeroone(\seq[m]{\alpha})$ expresses that zero or one
$\alpha_i$ valuates to $\top$. 
Further, we introduce the auxiliary constraint 
\begin{align*}
\enc{g(\seq[m]{t}) \in \termsbelow[F]} \defsym \bigvee_{f \in F} \enc{f \sp g} \wedge \bigwedge_{j=1}^m \enc{t_j \in \termsbelow[F]}
\end{align*}
and $\enc{x \in \termsbelow[F]} \defsym \top$ for $x \in \VS$.
Using these, clause $\cpop{ia}$ becomes expressible as 
\begin{multline*}
  \enc{f(\seq{s}) \cpop{ia} g(\seq[m]{t})} \defsym
  \enc{f \in \DS}
  \wedge \enc{f \sp g} \\
  \wedge \bigwedge_{j=1}^m (\esafe{g}{j} \imp \enc{s \gpop t_j})
  \wedge \bigwedge_{j=1}^m (\neg \esafe{g}{j} \imp \enc{s \gsq t_j}) \\
  \wedge \ezeroone(\seq[m]{\alpha}) 
  \wedge \bigwedge_{j=1}^m (\neg \alpha_j \imp \enc{t_j \in \termsbelow[\Fun(s)]}) \tpkt
\end{multline*}
Here $\enc{f \in \DS} = \top$ if $f \in \DS$ and otherwise $\enc{f \in \DS} = \bot$.
The propositional formula $\enc{s \gsq t}$ expresses the orientation with the $\gsq$ and is given by
\begin{align*}
  \enc{f(\seq{s}) \gsq t} \defsym \enc{f(\seq{s}) \csq{st} t} \vee \enc{f(\seq{s}) \csq{ia} t}
\end{align*}
and otherwise $\enc{x \gsq t} = \bot$, where
\begin{align*}
  \enc{f(\seq{s}) \csq{st} t} & \defsym \bigor_{i=1}^n ((\enc{s_i \gsq t} \vee \enc{s_i \eqis t}) \wedge (\enc{f \in \DS} \imp \neg \esafe{f}{i})) \\
  \enc{f(\seq{s}) \csq{ia} t} & \defsym 
  \begin{cases}
    \enc{f \in \DS} \wedge \enc{f \sp g} & \text{ if $t = g(\seq[m]{t})$} \\
    \quad \wedge \bigwedge_{j=1}^m \enc{f(\seq{s}) \gsq t_j} \\
    \bot & \text{ if $t \in \VS$}.
  \end{cases}
\end{align*}
This concludes the propositional formulation of clause $\cpop{ia}$.

The main challenge in formulating clause $\cpop{ep}$
is to deal with the encoding of multiset-comparisons. 
We proceed as in~\cite{SK07} and encode the underlying \emph{multiset cover}.
\begin{definition}
Let $\succ_\mul$ denote the multiset extension of a binary relation ${\succcurlyeq} = {\succ} \uplus {\eqi}$.
Then a pair of mapping $(\gamma, \varepsilon)$ 
where $\ofdom{\gamma}{\set{1,\dots,m} \to \set{1,\dots,n}}$ 
and $\ofdom{\varepsilon}{\set{1,\dots,n} \to \set{\top,\bot}}$
is a multiset cover on multisets $\mset{\seq{a}}$ and $\mset{\seq[m]{b}}$
if the following holds for all $j \in \{1,\dots,m\}$:
\begin{enumerate}\label{d:mscover}
\item\label{d:mscover:1} if $\gamma(j) = i$ then $a_i \succcurlyeq b_j$, in this case we say that $a_i$ \emph{covers} $b_j$; 
\item\label{d:mscover:2} if $\varepsilon(j) = \top$ then $s_{\tau(j)} \eqi t_j$ and $\tau$ is invective on $\{j\}$, 
  i.e., $a_{\tau(j)}$ covers only $b_j$.
\end{enumerate}
The multiset cover $(\gamma, \varepsilon)$ is said to be \emph{strict} if at least one cover is strict, 
i.e., $\varepsilon(j) = \bot$ for some $j \in \{1,\dots,m\}$.
\end{definition}
\noindent It is straight forward to verify that multiset covers characterise the multiset extension
of $\succ$ in the following sense.
\begin{lemma}
  We have $\mset{\seq{a}} \mextension{\succcurlyeq} \mset{\seq[m]{b}}$ if and only if there 
  exists a multiset cover $(\gamma, \varepsilon)$ on $\mset{\seq{a}}$ and $\mset{\seq[m]{b}}$.
  Moreover, $\mset{\seq{a}} \mextension{\succ} \mset{\seq[m]{b}}$ if and only if the cover is strict.
\end{lemma}

Consider the orientation $f(\seq{s}) \cpop{ep} g(\seq[m]{t})$. 
Then normal arguments are strictly, and safe arguments weakly decreasing with 
respect to the multiset-extension of $\gpop$. 
Since the partitioning of normal and safe argument is not fixed, 
in the encoding of $\cpop{ep}$ we formalise a multiset-comparison on \emph{all} arguments, 
where the underlying multiset-cover $(\gamma, \varepsilon)$ 
will be restricted so that if $s_i$ covers $t_j$, i.e., $\gamma(i) = j$, 
then both $s_i$ and $t_j$ are safe or respectively normal.
To this extend, for a specific multiset cover $(\gamma, \varepsilon)$ we introduce variables $\gamma_{i,j}$ and 
$\varepsilon_i$, where $\gamma_{i,j} = \top$ represents $\gamma(j) = i$ and
$\varepsilon_i = \top$ denotes $\varepsilon(i) = \top$ ($1 \leqslant i \leqslant n$, $1 \leqslant j \leqslant m$).
We set
\begin{multline*}
  \enc{f(\seq{s}) \cpop{ep} g(\seq[m]{t})} \defsym 
  \enc{f \in \DS}
  \wedge \enc{f \sp g} \\
  \wedge 
  \bigwedge_{i=1}^{n} \bigwedge_{j=1}^{m} \Bigl( \gamma_{i,j} \to \bigl( \varepsilon_i \to \enc{s_i \eqis t_j} \bigr)
                                             \wedge \bigl( \neg \varepsilon_i \to \enc{s_i \gpop t_j} \bigr)
                                             \wedge \bigl( \esafe{f}{i} \iff \esafe{g}{j} \bigr)
                                      \Bigr) \\
  \wedge \bigwedge_{j=1}^m \eone(\gamma_{1,j},\dots,\gamma_{n,j}) 
  \wedge \bigwedge_{i=1}^{n} \bigl(\varepsilon_i \to \eone(\gamma_{i,1},\dots,\gamma_{i,m})\bigr) 
  \wedge  \bigvee_{i=1}^n \bigl( \neg \esafe{f}{i} \wedge \neg \varepsilon_i \bigr)  \tpkt
\end{multline*}
Here the first line establishes the Condition~\eref{d:mscover}{1}, where
$\esafe{f}{i} \iff \esafe{g}{j}$ additionally enforces the separation of normal from safe arguments.
The final line formalises
that $\gamma$ maps $\{1,\dots,m\}$ to $\{1,\dots,n\}$, Condition~\eref{d:mscover}{2}
as well as the strictness condition on normal arguments.
This completes the encoding of $\gpop$.

\begin{lemma}
  Suppose $\mu$ induces an admissible precedence $\qp$ and satisfies $\enc{s \gpop t}$. 
  Then $s \gpop t$ with respect to the precedence $\qp$.
  Vice versa, if $s \gpop t$ then $\enc{s \gpop t}$ is satisfiable assignments $\mu$ 
  that induce the precedence underlying $\gpop$.
\end{lemma}

As a predicative recursive TRS $\RS$ is a constructor TRS compatible with some 
polynomial path order $\gpop$, putting the constraints together we get the following theorem.

\begin{theorem}
  Let $\RS$ be a constructor TRS.\@ 
  The propositional formula 
  $$
  \vpredrec(\RS) \defsym \vprec(\DS) \wedge \bigwedge_{{l \to r} \in \RS} \enc{l \gpop r}
  $$
  is satisfiable if and only if $\RS$ is predicative recursive.
\end{theorem}

We have implemented this reduction to \SAT~in our complexity analyser \TCT.\@
As underlying \SAT-solver we employ the open source solver \minisat~\cite{ES03}.
On the example from the introduction, \TCT\ outputs 
the following result in a fraction of a second.
\small{
\begin{verbatim}
  The input was oriented with 'POP*' as induced by the precedence

  member > if, member > eq, guess > choice, consistent > if,
   consistent > member, consistent > neg, sat > guess, sat > sat',
   sat' > if, sat' > consistent .
  
  Oriented rules in predicative notation are as follows.

    sat'(cnf, assign;) -> if(; consistent(assign;), assign, unsat())
    sat(cnf;) -> sat'(cnf, guess(cnf;);)
    consistent(cons(; l, ls);) ->
       if(; member(ls; neg(l;)), ff(), consistent(ls;))
    consistent(nil();) -> tt()
    guess(nil();) -> nil()
    guess(cons(; c, cs);) ->
       cons(; choice(c;), guess(cs;))
    choice(cons(; a, nil());) -> a
    choice(cons(; a, cons(; b, bs));) -> a
    choice(cons(; a, cons(; b, bs));) -> choice(cons(; b, bs);)
    neg(1(; x);) -> 0(; x)
    neg(0(; x);) -> 1(; x)
    eq(1(; y); 1(; x)) -> eq(y; x)
    eq(0(; y); 1(; x)) -> ff()
    eq(1(; y); 0(; x)) -> ff()
    eq(0(; y); O(; x)) -> eq(y; x)
    eq(e(); e()) -> tt()
    member(cons(; y, ys); x) -> if(; eq(y; x), tt(), member(ys; x))
    member(nil(); x) -> ff()
    if(; ff(), t, e) -> e
    if(; tt(), t, e) -> t
\end{verbatim}
}

\paragraph{\textbf{Efficiency Considerations}}
The \SAT-solver \minisat\ requires its input in CNF.\@
For a concise translation of $\vpredrec(\RS)$ to CNF 
we use the approach of Plaisted and Greenbaum~\cite{PG86} that 
gives an equisatisfiable CNF linear in size.
Our implementation also eliminates redundancies resulting from 
multiple comparisons of the same pair of term $s, t$ by 
replacing subformulas $\enc{s \gpop t}$ with unique 
propositional atoms $\delta_{s,t}$. Since $\enc{s \gpop t}$ 
occurs only in positive contexts, it suffices to 
add $\delta_{s,t} \imp \enc{s \gpop t}$, resulting in an equisatisfiable formula.
Also during construction of $\vpredrec(\RS)$ our implementation
performs immediate simplifications under Boolean laws.


\section{Experimental Assessment}\label{s:exps}

In this section we present an empirical evaluation of polynomial path orders.
We selected two testbeds: Testbed~\textsf{TC} constitutes 
of 597 terminating constructor TRSs, obtained
by restricting the innermost runtime complexity problemset 
from the \emph{termination problem database} (\emph{TPDB} for short), version 8.0
to known to be terminating constructor TRSs.
Termination is checked against the full run of the complexity competition from December 2011
Testbed~\textsf{TCO}, containing 290 examples, results from restricting Testbed~\textsf{TC} to 
orthogonal systems.
Unarguably the TPDB is an imperfect choice as examples were collected primarily to 
assess the strength of termination provers, but it is at the moment the only 
extensive source of TRSs. Since the creation of the dedicated 
complexity categories in 2008 the situation, although slowly, changes to the better.

Experiments were conducted with $\TCT$ version 1.9.1%
\footnote{Available from \url{http://cl-informatik.uibk.ac.at/software/tct/projects/tct/archive/}.},
on a laptop with 4Gb of RAM and Intel${}^\text{\textregistered}$ Core${}^\text{\texttrademark}$ i7--2620M CPU (2.7GHz, quad-core).
We assess the strength of $\POPSTAR$ and $\POPSTARP$ in comparison to its predecessors $\MPO$ and $\LMPO$.\@
The implementation of $MPO$ and $\LMPO$ follows the line of polynomial path orders 
as explained in Section~\ref{s:impl}.
We contrast these syntactic techniques to \emph{interpretations}
as implemented in our complexity tool $\TCT$.\@
The last column show result of constructor restricted 
matrix interpretations~\cite{MMNWZ11} (dimension $1$ and $3$)
as well as polynomial interpretations~\cite{BCMT01} (degree $2$ and $3$), 
run in parallel on the quad-core processor.
We employ interpretations in their default configuration of \TCT, 
noteworthy coefficients (respectively entries in coefficients) 
range between $0$ and $7$, and we also make use of the \emph{usable argument positions} 
criterion~\cite{HM11} that weakens monotonicity constraints.
Table~\ref{tbl:exp1}%
\footnote{Full evidence available at \url{http://cl-informatik.uibk.ac.at/software/tct/experiments/popstar}.}
shows totals on systems that can respectively cannot be handled.
To the right of each entry we annotate the average execution time, in seconds.

\newcommand{\tm}[1]{\parbox[b]{9mm}{\bf{\tiny{$\backslash$#1}}}}
\renewcommand{\c}[1]{\parbox[b]{9mm}{{\hfill\small{#1}}}}
\begin{table}[h]
  \centering
  \begin{tabular}{l@{}l|cccc|c}
    \TOP & 
    & \MPO 
    & \LMPO
    & \POPSTAR
    & \POPSTARP
    & interpretations
    \BOT
    \\
    \hline\hline
    \textbf{TC} \TOP 
    & \textsf{compatible}
    & \c{76}\tm{0.33} 
    & \c{57}\tm{0.20} 
    & \c{43}\tm{0.18} 
    & \c{56}\tm{0.19} 
    & \c{139}\tm{2.77} 
    \\
    & \TOP \textsf{incompatible}
    & \c{521}\tm{0.58} 
    & \c{540}\tm{0.47} 
    & \c{554}\tm{0.42} 
    & \c{541}\tm{0.43} 
    & \c{272}\tm{6.47} 
    \\
    & \TOP\BOT \textsf{timeout}
    & --- 
    & --- 
    & --- 
    & --- 
    & \c{186}\tm{25.0} 
    \\
    \hline\hline
    \textbf{TCO} \TOP 
    & \textsf{compatible}
    & \c{40}\tm{0.29} 
    & \c{29}\tm{0.16} 
    & \c{24}\tm{0.14} 
    & \c{29}\tm{0.15} 
    & \c{75}\tm{2.81} 
    \\
    & \TOP \textsf{incompatible}
    & \c{250}\tm{0.33} 
    & \c{261}\tm{0.27} 
    & \c{266}\tm{0.26} 
    & \c{261}\tm{0.27} 
    & \c{133}\tm{6.12} 
    \\
    & \TOP\BOT \textsf{timeout}
    & --- 
    & --- 
    & --- 
    & --- 
    & \c{82}\tm{25.0} 
    \\
    \hline\hline
  \end{tabular}
\caption{Empirical Evaluation, comparing syntactic to semantic techniques.}
\label{tbl:exp1}
\end{table}

It is immediate that syntactic techniques cannot compete with the expressive 
power of interpretations.
In Testbed~\textsf{TC} there are in fact only three examples 
compatible with \POPSTARP\ where \TCT~could not find interpretations.
There are additionally four examples compatible with \LMPO\ but not so with interpretations, 
including the TRS $\RSbin$ from Example~\ref{ex:RS2}. 
All but one (noteworthy the merge-sort algorithm from Steinbach and K\"uhlers collection 
\cite[Example~2.43]{SK90}) 
of these do in fact admit exponential runtime-complexity, 
thus a~priori they are not compatible to the restricted interpretations.
We emphasise that parameter substitution significantly increases the strength of 
\POPSTAR, 13 examples are provable by \POPSTARP\ but neither by \POPSTAR\ nor \LMPO.\@
\LMPO\ could benefit from parameter substitution, 
we conjecture that the resulting order is still sound for $\FP$. 

On Testbed~\textsf{TCO}, containing only orthogonal TRSs, in total 75 systems (26\% of the testbed)
can be verified to encode polytime computable functions, 35 (12\% of the testbed)
can be verified polytime computable by only syntactic techniques. 
It should be noted that not all examples appearing in our collection encode polytime computable 
functions, the total amount of such systems is unknown. 

Table~\ref{tbl:exp1} clearly illustrates one of our main motivations for investigating
syntactic techniques. 
Our complexity analyser \TCT\ recursively decomposes complexity problems using 
various complexity preserving transformation techniques, 
discarding those problems that can be handled by basic techniques as 
contrasted in Table~\ref{tbl:exp1}.
Certificates are only obtained
if finally all subproblems can be discarded,
above all it is crucial that subproblems can be discarded 
quickly. 
\POPSTARP\ succeeds on average 14 times faster than polynomial and 
matrix interpretations run parallel, it can be safely 
preposed to interpretations, speeding up the overall procedure.
Note that the difficulty of implementing interpretations efficiently 
is also reflected in the total number of timeouts.


\section{Conclusion and Future Work}\label{s:conclusion}
We propose a new order, the polynomial path order $\POPSTAR$. 
The order $\POPSTAR$ is a syntactical restriction of multiset path orders, 
with the distinctive feature that the (innermost) runtime complexity
of compatible TRSs lies in $O(n^d)$ for some $d$.
Based on $\POPSTAR$, we delineate a class of rewrite systems, dubbed
systems of predicative recursion, 
so that the class of functions computed by these systems
corresponds to $\FP$, the class of polytime computable functions.
We have shown that an extension of $\POPSTAR$, the order $\POPSTARP$
that also accounts for parameter substitution, 
increases the intensionality of $\POPSTAR$.
In contrast to interpretations, 
$\POPSTAR$ is partly lacking in intensionality but surpluses
in verification time.

In our complexity prover \TCT, we do not intend to replace 
semantic techniques, but rather prepose them by \POPSTARP, in 
order to improve \TCT\ both in analytic power and speed. 
With \TCT\ we are in particular interested in obtaining asymptotically 
tight bounds. 
Although we could estimate the degree of the witnessing
bounding function for \POPSTAR\ and \POPSTARP, 
a bound extracted from our proof yields unnecessarily an overestimation, 
compare Theorem~\ref{t:pop} and particular the preceding construction of the degree $d_{k,p}$.
Partly this is due to the underlying multiset extension.
Future investigations will certainly include establishing 
tighter bounds.


\section*{Acknowledgement}
We are in particular thankful to Nao Hirokawa for fruitful discussions.

\bibliographystyle{plain}

\begin{thebibliography}{10}

\bibitem{AAGGPRRZ:2009}
E.~Albert, P.~Arenas, S.~Genaim, M.~G{\'o}mez-Zamalloa, G.~Puebla,
  D.~Ram\'{\i}rez, G.~Rom{\'a}n, and D.~Zanardini.
\newblock {Termination and Cost Analysis with COSTA and its User Interfaces}.
\newblock {\em Electronic Notes in Theoretical Computer Science},
  258(1):109--121, 2009.

\bibitem{AM05}
T.~Arai and G.~Moser.
\newblock {Proofs of Termination of Rewrite Systems for Polytime Functions}.
\newblock In {\em Proc.\ of the 25th FSTTCS}, volume 3821 of {\em Lecture Notes
  in Computer Science}, pages 529--540. Springer Verlag, 2005.

\bibitem{AG00}
T.~Arts and J.~Giesl.
\newblock {Termination of Term Rewriting using Dependency Pairs}.
\newblock {\em Theoretical Computer Science}, 236(1--2):133--178, 2000.

\bibitem{AM08}
M.~Avanzini and G.~Moser.
\newblock {Complexity Analysis by Rewriting}.
\newblock In {\em Proc.\ of 9th FLOPS}, volume 4989 of {\em Lecture Notes in
  Computer Science}, pages 130--146. Springer Verlag, 2008.

\bibitem{AM09}
M.~Avanzini and G.~Moser.
\newblock {Dependency Pairs and Polynomial Path Orders}.
\newblock In {\em Proc.\ of 20th RTA}, volume 5595 of {\em Lecture Notes in
  Computer Science}, pages 48--62. Springer Verlag, 2009.

\bibitem{AM09b}
M.~Avanzini and G.~Moser.
\newblock {Polynomial Path Orders and the Rules of Predicative Recursion with
  Parameter Substitution}.
\newblock In {\em Proc.\ of the 10th WST}, pages 16--20, 2009.

\bibitem{AM10b}
M.~Avanzini and G.~Moser.
\newblock {Closing the Gap Between Runtime Complexity and Polytime
  Computability}.
\newblock In {\em Proc.\ of the 21th RTA}, volume~6 of {\em Leibniz
  International Proceedings in Informatics}, pages 33--48, 2010.

\bibitem{AM10}
M.~Avanzini and G.~Moser.
\newblock {Complexity Analysis by Graph Rewriting}.
\newblock In {\em Proceedings of the 10th FLOPS}, volume 6009 of {\em Lecture
  Notes in Computer Science}, pages 257--271. Springer Verlag, 2010.

\bibitem{AMS08}
M.~Avanzini, G.~Moser, and A.~Schnabl.
\newblock Automated implicit computational complexity analysis (system
  description).
\newblock In {\em Proc.\ of 4th IJCAR}, volume 5195 of {\em Lecture Notes in
  Computer Science}, pages 132--139. Springer Verlag, 2008.

\bibitem{BN98}
F.~Baader and T.~Nipkow.
\newblock {\em {Term Rewriting and All That}}.
\newblock Cambridge University Press, 1998.

\bibitem{BMR09}
P.~Baillot, J.-Y. Marion, and S.~Ronchi~Della Rocca.
\newblock {Guest Editorial: Special Issue on Implicit Computational
  Complexity}.
\newblock {\em ACM Transactions on Computational Logic}, 10(4), 2009.

\bibitem{BW96}
A.~Beckmann and A.~Weiermann.
\newblock {A Term Rewriting Characterization Of the Polytime Functions and
  Related Complexity Classes}.
\newblock {\em Archive for Mathematical Logic}, 36:11--30, 1996.

\bibitem{BC92}
S.~Bellantoni and S.~Cook.
\newblock {A new Recursion-Theoretic Characterization of the Polytime
  Functions}.
\newblock {\em Computational Complexity}, 2(2):97--110, 1992.

\bibitem{BCMT01}
G.~Bonfante, A.~Cichon, J.-Y. Marion, and H.~Touzet.
\newblock {Algorithms with Polynomial Interpretation Termination Proof}.
\newblock {\em Journal of Functional Programming}, 11(1):33--53, 2001.

\bibitem{BMOG12}
M.~Brockschmidt, R.~Musiol, C.~Otto, and J.~Giesl.
\newblock {Automated Termination Proofs for Java Programs with Cyclic Data}.
\newblock In {\em Proc. of 24th CAV}, volume 7358 of {\em Lecture Notes in
  Computer Science}, pages 185--122. Springer Verlag, 2012.

\bibitem{B95}
W.~Buchholz.
\newblock {Proof-theoretical Analysis of Termination Proofs}.
\newblock {\em Annals of Pure and Applied Logic}, 75:57--65, 1995.

\bibitem{CW97}
E.~A. Cichon and A.~Weiermann.
\newblock {Term Rewriting Theory for the Primitive Recursive Functions}.
\newblock {\em Annals of Pure and Applied Logic}, 83(3):199--223, 1997.

\bibitem{LM09}
U.~{Dal Lago} and S.~Martini.
\newblock {On {C}onstructor {R}ewrite {S}ystems and the {L}ambda-{C}alculus}.
\newblock In {\em Proc.\ of 36th ICALP}, volume 5556 of {\em Lecture Notes in
  Computer Science}, pages 163--174. Springer Verlag, 2009.

\bibitem{ES03}
Niklas E{\'e}n and Niklas S{\"o}rensson.
\newblock {An Extensible SAT-solver}.
\newblock In {\em Proc.\ of 6th SAT}, volume 2919 of {\em Lecture Notes in
  Computer Science}, pages 502--518. Springer Verlag, 2003.

\bibitem{Ferreira95}
M.~C.~F. Ferreira.
\newblock {\em {Termination of Term Rewriting}}.
\newblock PhD thesis, University of Utrecht, November 1995.
\newblock Well-foundedness, Totality and Transformations.

\bibitem{GMC09}
S.~Gulwani, K.K. Mehra, and T.M. Chilimbi.
\newblock {SPEED: Precise and Efficient Static Estimation of Program
  Computational Complexity}.
\newblock In {\em Proc.\ of 36th POPL}, pages 127--139. Association for
  Computing Machinery, 2009.

\bibitem{HM08}
N.~Hirokawa and G.~Moser.
\newblock {Automated Complexity Analysis Based on the Dependency Pair Method}.
\newblock In {\em Proc.\ of the 4th IJCAR}, volume 5195 of {\em Lecture Notes
  in Artificial Inteligence}, pages 364--380. Springer Verlag, 2008.

\bibitem{HM11}
N.~Hirokawa and G.~Moser.
\newblock {Automated Complexity Analysis Based on the Dependency Pair Method}.
\newblock {\em CoRR}, abs/1102.3129, 2011.
\newblock submitted.

\bibitem{H92}
D.~Hofbauer.
\newblock {Termination Proofs by Multiset Path Orderings Imply Primitive
  Recursive Derivation Lengths}.
\newblock {\em Theoretical Computer Science}, 105:129--140, 1992.

\bibitem{HL89}
D.~Hofbauer and C.~Lautemann.
\newblock {Termination Proofs and the Length of Derivations}.
\newblock In {\em Proc.\ of 3rd RTA}, volume 355 of {\em Lecture Notes in
  Computer Science}, pages 167--177. Springer Verlag, 1989.

\bibitem{HAH11}
J.~Hoffmann, K.~Aehlig, and M.~Hofmann.
\newblock {Multivariate Amortized Resource Analysis}.
\newblock In {\em Proc.\ of 38th POPL}, pages 357--370. Association for
  Computing Machinery, 2011.

\bibitem{GRMSST11}
J\"{u}rgen J.~Giesl, M.~Raffelsieper, P.~Schneider-Kamp, S.~Swiderski, and
  R.~Thiemann.
\newblock {Automated Termination Proofs for Haskell by Term Rewriting}.
\newblock {\em ACM Transactions on Programming Languages and Systems},
  33:1--39, 2011.

\bibitem{L91}
D.~Leivant.
\newblock {A Foundational Delineation of Computational Feasiblity}.
\newblock In {\em Proc.\ of 6ht LICS}, pages 2--11. IEEE Computer Society,
  1991.

\bibitem{M03}
J.-Y. Marion.
\newblock {Analysing the Implicit Complexity of Programs}.
\newblock {\em Information and Computation}, 183:2--18, 2003.

\bibitem{MMNWZ11}
A.~Middeldorp, G.~Moser, F.~Neurauter, J.~Waldmann, and H.~Zankl.
\newblock {Joint Spectral Radius Theory for Automated Complexity Analysis of
  Rewrite Systems}.
\newblock In {\em Proc.\ of 4th CAI}, volume 6472 of {\em Lecture Notes in
  Computer Science}, pages 1--20. Springer Verlag, 2011.

\bibitem{M06}
G.~Moser.
\newblock {Derivational Complexity of Knuth-Bendix Orders Revisited}.
\newblock In {\em Proc.\ of the 13th LPAR}, volume 4246 of {\em Lecture Notes
  in Artificial Inteligence}, pages 75--89. Springer Verlag, 2006.

\bibitem{M09}
G.~Moser.
\newblock {Proof Theory at Work: Complexity Analysis of Term Rewrite Systems}.
\newblock {\em CoRR}, abs/0907.5527, 2009.
\newblock Habilitation Thesis.

\bibitem{MS11}
G.~Moser and A.~Schnabl.
\newblock {The Derivational Complexity Induced by the Dependency Pair Method}.
\newblock {\em Logical Methods in Computer Science}, 7(3), 2011.

\bibitem{NEG11}
L.~Noschinski, F.~Emmes, and J.~Giesl.
\newblock {A Dependency Pair Framework for Innermost Complexity Analysis of
  Term Rewrite Systems}.
\newblock In {\em Proc.\ of 23rd CADE}, Lecture Notes in Computer Science,
  pages 422--438. Springer Verlag, 2011.

\bibitem{O01}
E.~Ohlebusch.
\newblock {Termination of Logic Programs: Transformational Methods Revisited}.
\newblock {\em Applicable Algebra in Engineering, Communication and Computing},
  12(1/2):73--116, 2001.

\bibitem{OBEG10}
C.~Otto, M.~Brockschmidt, C.~v.~Essen, and J.~Giesl.
\newblock {Automated Termination Analysis of Java Bytecode by Term Rewriting}.
\newblock In {\em Proc.\ of 21th RTA}, pages 259--276, 2010.

\bibitem{Papa}
Christos~H. Papadimitriou.
\newblock {\em {C}omputational {C}omplexity}.
\newblock {A}ddison {W}esley {L}ongman, second edition, 1995.

\bibitem{PG86}
D.~A. Plaisted and S.~Greenbaum.
\newblock {A Structure-Preserving Clause Form Translation}.
\newblock {\em Journal of Symbolic Computation}, 2(3):293--304, 1986.

\bibitem{SFTGACMZ07}
P.~Schneider-Kamp, C.~Fuhs, R.~Thiemann, J.~Giesl, E.~Annov, M.~Codish,
  A.~Middeldorp, and H.~Zankl.
\newblock Implementing {RPO} and {POLO} {U}sing {SAT}.
\newblock In {\em DDP}, number 07401 in Leibniz International Proceedings in
  Informatics. Dagstuhl, 2007.

\bibitem{SK07}
P.~Schneider-Kamp, R.~Thiemann, E.~Annov, M.~Codish, and J.~Giesl.
\newblock {Proving Termination Using Recursive Path Orders and SAT Solving}.
\newblock In {\em Proc.\ of 6th FroCoS}, volume 4720 of {\em Lecture Notes in
  Computer Science}, pages 267--282. Springer Verlag, 2007.

\bibitem{Simmons:1988}
H.~Simmons.
\newblock {The Realm of Primitive Recursion}.
\newblock {\em Applied Mathematicas Letters}, 27:177--188, 1988.

\bibitem{SK90}
J.~Steinbach and U.~K{\"u}hler.
\newblock {Check your Ordering - Termination Proofs and Open Problems}.
\newblock Technical Report SEKI-Report SR-90-25, University of Kaiserslautern,
  1990.

\bibitem{SSG10}
T.~Str{\"o}der, P.~Schneider-Kamp, and J.~Giesl.
\newblock {Dependency Triples for Improving Termination Analysis of Logic
  Programs with Cut}.
\newblock In {\em Proc. of 20th LOPSTR}, Lecture Notes in Computer Science,
  pages 184--199. Springer Verlag, 2010.

\bibitem{W95}
A.~Weiermann.
\newblock {Termination Proofs for Term Rewriting Systems with Lexicographic
  Path Orderings Imply Multiply Recursive Derivation Lengths}.
\newblock {\em Theoretical Computer Science}, 139(1,2):355--362, 1995.

\bibitem{ZM07}
H.~Zankl and A.~Middeldorp.
\newblock {Satisfying KBO Constraints}.
\newblock In {\em {Proc.\ of the 18th RTA}}, volume 4533 of {\em Lecture Notes
  in Computer Science}, pages 389--403. Springer Verlag, 2007.

\bibitem{HZMK10}
Harald Zankl and Martin Korp.
\newblock {Modular Complexity Analysis via Relative Complexity}.
\newblock In {\em Proc.\ of 21st RTA}, volume~6 of {\em Leibniz International
  Proceedings in Informatics}, pages 385--400, 2010.

\bibitem{Z95}
H.~Zantema.
\newblock {Termination of Term Rewriting by Semantic Labelling}.
\newblock {\em Fundamenta Informaticae}, 24(1/2):89--105, 1995.

\bibitem{ZulegerGSV11}
F.~Zuleger, S.~Gulwani, M.~Sinn, and H.~Veith.
\newblock {Bound Analysis of Imperative Programs with the Size-Change
  Abstraction}.
\newblock In {\em Proc.\ of 18th SAS}, volume 6887 of {\em Lecture Notes in
  Computer Science}, pages 280--297. Springer Verlag, 2011.

\end{thebibliography}

\end{document}